\documentclass[12pt,a4paper]{article}
\usepackage{amsmath}
\usepackage{amsfonts}
\usepackage{amssymb}
\usepackage{amsthm}
\usepackage[mathscr]{euscript}
\usepackage{bbm}
\usepackage{bm}
\usepackage{mathrsfs}
\usepackage{prodint}
\usepackage{graphicx}
\usepackage[left=0.75in,right=0.75in]{geometry}
\usepackage{setspace}
\usepackage{hyperref}
\usepackage[width=0.9\textwidth]{caption}

\usepackage{xcolor}

\usepackage{authblk}
\usepackage{paralist}

\graphicspath{{./images/}}

\usepackage{natbib}
\usepackage{multibib}
\newcites{supp}{References}

\usepackage{algorithm}
\usepackage{algorithmicx}
\usepackage[noend]{algpseudocode}

\newcommand{\ourtitle}{{Sequentially Doubly Robust Estimation of Conditional Survival Probability with Time-Varying Covariates}}
\title{\ourtitle}

\date{}
\author[1]{Hongxiang Qiu}
\author[2,3]{Marco Carone}
\author[4,5]{Alex Luedtke}
\author[2,3,6]{Peter B. Gilbert}

\affil[1]{Department of Epidemiology and Biostatistics, Michigan State University}
\affil[2]{Department of Biostatistics, University of Washington}
\affil[3]{Vaccine and Infectious Disease Division, Fred Hutchinson Cancer Center}
\affil[4]{Department of Health Care Policy, Harvard Medical School}
\affil[5]{Department of Statistics, University of Washington}
\affil[6]{Public Health Sciences Division, Fred Hutchinson Cancer Center}

\allowdisplaybreaks
\newtheorem{theorem}{Theorem}
\newtheorem{lemma}{Lemma}
\newtheorem{corollary}{Corollary}

\theoremstyle{definition}
\newtheorem{remark}{Remark}
\newtheorem{condition}{Condition}

\DeclareMathOperator{\expit}{expit}

\newcommand{\real}{{\mathbb{R}}}

\newcommand{\ind}{{\mathbbm{1}}}
\newcommand{\expect}{{\mathbb{E}}}

\newcommand{\intd}{{\mathrm{d}}}

\newcommand{\smallo}{{\mathrm{o}}}

\newcommand{\G}{{\mathrm{G}}}
\newcommand{\IPCW}{{\mathrm{IPCW}}}
\newcommand{\SDR}{{\mathrm{SDR}}}
\newcommand{\T}{{\mathscr{T}}}
\newcommand{\TS}{{\mathscr{C}}}
\newcommand{\const}{{\mathrm{C}}}
\newcommand{\CDE}{{\mathrm{CDE}}}

\newcommand\independent{\protect\mathpalette{\protect\independenT}{\perp}}
\def\independenT#1#2{\mathrel{\rlap{$#1#2$}\mkern2mu{#1#2}}}

\begin{document}

\maketitle

\begin{abstract}
    It is often of interest to study the association between covariates and the cumulative incidence of a right-censored time-to-event outcome.
    When time-varying covariates are measured on a fixed discrete time scale, it is desirable to account for more up-to-date covariates that may inform future censoring.
    Existing methods rely on stringent parametric assumptions, do not account for informative censoring due to time-varying covariates for continuous time, estimate a marginal survival probability, or do not fully use the discrete-time structure of time-varying covariates.
    We propose a nonparametric and sequentially doubly robust estimator of the continuous-time survival probability conditional on covariates in this setting.
    We also show that, in the special case of estimating the marginal survival probability, our estimator is asymptotically efficient.
    We demonstrate the superior performance of our estimator in a simulation experiment, and apply the method to a COVID-19 vaccine efficacy trial.

\end{abstract}

\section{Introduction} \label{sec: intro}

In many studies such as cohort studies and clinical trials, the outcome of interest is a continuous time-to-event subject to right-censoring. Participants attend study visits regularly on a fixed discrete time scale and have covariates measured at these visits.
In such settings, an important scientific question is to estimate the association between baseline or time-varying covariates and the survival probabilities at particular time points.
For example, in vaccine trials, participants receive a few doses and visit clinics at predetermined times.
At these visits, time-varying covariates such as immune response biomarkers are also measured.
It is of interest to estimate the disease-free survival time conditional on these covariates, because identifying biomarkers that strongly associate with disease risk has many applications such as 
validating surrogate endpoints. For vaccines, there is a large literature on the use of immune response biomarkers as surrogate endpoints for disease endpoints, with many applications including shedding light on mechanisms of protection, accelerating approval of vaccines, and accelerating iterative improvement of vaccines \protect\citep[e.g.,][]{Plotkin2008,gilbert2022covid}. Moreover, the efficacy of a vaccine often wanes, in which case the immune response surrogate generally only applies for disease occurrence within a limited period of follow-up after vaccination. By providing results on how the conditional survival curve changes with time since vaccination, our methods directly address this issue by characterizing how the biomarker-disease association changes with time.
One major challenge in such scenarios is potential informative censoring due to post-vaccination time-varying covariates.

Our motivating question arises from a COVID-19 vaccine efficacy trial \protect\citep{Dayan2023}. The scientific question of interest concerns the association between antibody levels and the incidence of COVID-19 occurrence adjusting for baseline covariates in a general population.
As participants had time-varying antibody levels measured at visits after receiving two doses of vaccines or placebo, it is useful to account for these post-treatment covariates that might explain censoring afterwards, even though the scientific question is unrelated to these additional covariates.

Existing approaches to estimating survival probabilities with time-varying covariates often rely on strong parametric or semiparametric assumptions. For example, in the Cox proportional hazards model, the form of the hazard or intensity function is restricted to be proportional to time-varying covariates \protect\citep{Andersen2007,VanHouwelingen2011}. The additive hazards model also makes a strong semiparametric assumption that the additive effects of time-varying covariates on the hazard are constant over time \protect\citep{Aalen1980}. The traditional accelerated failure time model assumes a parametric model on the true distribution \protect\citep{Klein2003}. Although this approach also has semiparametric versions
\protect\citep[e.g.,][]{Prentice1978,Kalbfleisch2011},
the assumption on the covariate effects may still be stringent.
The extended hazard regression model is a more general model including the Cox proportional hazard model and the accelerated failure time model as special cases \protect\citep{Etezadi-Amoli1987,Tseng2014}, but the assumptions might fail in practice.
In addition, standard software implementing these methods can often only be applied when the dimension of time-varying covariates is fixed; however, in a trial setting, it is often desirable to utilize all historical covariates at each time point so that the dimension of time-varying covariates increases over time.

Beyond parametric and semiparametric methods, nonparametric methods have also been developed. These approaches allow for the use of almost any flexible method to estimate nuisances such as survival curves conditional on time-varying covariates \protect\citep{Hubbard2000,Rotnitzky2005,VanderLaan2002,VanderLaan1998}. Such flexibility is desirable because stringent parametric or semiparametric assumptions are not required, and such methods have become feasible with recent advances in machine learning techniques for censored data \protect\citep{Ishwaran2008,Katzman2018,Westling2023,Wolock2024}.
However, these methods focus on estimating a marginal survival probability, or more generally, a summary of the marginal survival curve, while accounting for time-varying covariates, and such methods cannot be directly used to estimate the survival probabilities conditional on continuous covariates.
A notable exception is \protect\citet{Hubbard2000}, where the authors briefly mentioned the estimation of conditional survival probabilities but provided no convergence or robustness results.

The above works also do not explicitly consider time-varying covariates collected on a fixed time scale and mainly focus on general time-varying covariates that may be measured continuously.
The discrete-time structure of time-varying covariates introduces additional convenience in nonparametric estimation compared to continuously measured time-varying covariates.
Similar methods have also been developed for the case with only baseline covariates and no time-varying covariates \protect\citep[e.g.,][]{Bai2013,Sjolander2017,Westling2023}.
For discrete times, \protect\citet{Rotnitzky2017} and \protect\citet{Luedtke2017} proposed sequentially doubly robust (SDR, a special case of multiply robust) methods to estimate survival probabilities conditional on covariates.
These methods fit a sequence of regressions backward in time. If either nuisance is correctly specified at the last stage, the regression at the second-to-last stage targets the truth; if either is also correctly specified at the second-to-last stage, the regression at the stage before that also targets the truth, and so on. The resulting estimators are therefore consistent if, at each stage, either nuisance is correctly specified, and which one that is may differ across stages.
Conventional doubly robust estimators are guaranteed to be consistent under only a subset of these patterns, most notably when all outcome regressions or all censoring mechanisms are correctly specified.

In this paper, we propose a flexible, nonparametric, and SDR (in the sense of Theorem~\ref{thm: SDR identify} and Corollary~\ref{coro: SDR}) estimator of survival probabilities conditional on covariates for continuous time-to-event, while accounting for informative right-censoring due to time-varying covariates measured on a fixed discrete time scale.
We also show that this estimator can be adapted to estimate a marginal survival probability and is asymptotically efficient under a nonparametric model.
In doing so, we derive novel identifications of conditional survival probabilities and a novel SDR transformation based on semiparametric efficiency theory \protect\citep[e.g.,][]{Pfanzagl1985,Pfanzagl1990,Newey1990,Bolthausen2002}.
The proposed estimators are implemented in an R package available at \url{https://github.com/QIU-Hongxiang-David/MRsurv}.
Our proposed estimators move beyond the discrete-time constraint in existing literature, offering more granular and robust estimation of conditional survival probabilities in dynamic longitudinal settings.

\section{Problem setup} \label{sec: setup}

Suppose that each individual's covariates are measured at fixed time points $0=: t_1 < \ldots < t_K$, which we will refer to as visit times.
Let $T>0$ denote the time-to-event of interest---such as disease, infection, progression, or death---and $C>0$ denote the time to right-censoring such as loss to follow-up from the trial or end of the trial.
Each of $T$ and $C$ may be continuous or discrete.
In other words, covariate information is updated at fixed visits $t_k$, while censoring and event processes may run continuously between adjacent visits $t_k$ and $t_{k+1}$.
Let $\tilde{L}_k \in \mathcal{L}_k$ denote the vector of covariates that would have been measured at visit time $t_k$ ($k=1,\ldots,K$) if censoring were not present. We set $\tilde{L}_k = 0$ if the event occurs before or at $t_k$.
These time-varying covariates may contain the most recent information about future censoring and events.
Throughout this paper, we use $a \wedge b$ to denote $\min\{a,b\}$ for $a,b \in \real$.
Let $X=T \wedge C$ be the time of follow-up and $\Delta= \ind(T \leq C)$ be the indicator of observing the event. We use $L_k := \ind(X > t_k) \tilde{L}_k$ to denote the covariates that are actually observed at $t_k$.
We also assume that censoring occurs before measuring covariates at $t_k$ if $C=t_k$. 
We define $\tilde{H}_k := (\tilde{L}_1,\ldots,\tilde{L}_k)$ and $H_k := (L_1,\ldots,L_k)$ ($k=1,\ldots,K$) to be the historical covariates that would have been measured if censoring were not present and are actually measured up to visit time $t_k$, respectively.
We define $\tilde{H}_0$ and $H_0$ to be $\emptyset$.
Finally, a prototypical observation data point is defined as $O:=(H_K,X,\Delta)$. We assume that an independent and identically distributed sample consisting of observation units $O_1,\ldots,O_n$ is drawn from a true distribution $P_*$ in a nonparametric model.
We illustrate this study design and data structure in Fig.~\ref{fig: data structure}.
We use $\tilde{O} := (\tilde{H}_K,T,C)$ to denote a prototypical full data point and assume that $\tilde{O}_1,\ldots,\tilde{O}_n$ are independent draws from $\tilde{P}_*$.
Let $\mathcal{H}_k$ denote the support of $\tilde{H}_k$ under $\tilde{P}_*$ for $k=1,\ldots,K$.

\begin{figure}
    \centering
    \includegraphics[width=0.6\linewidth]{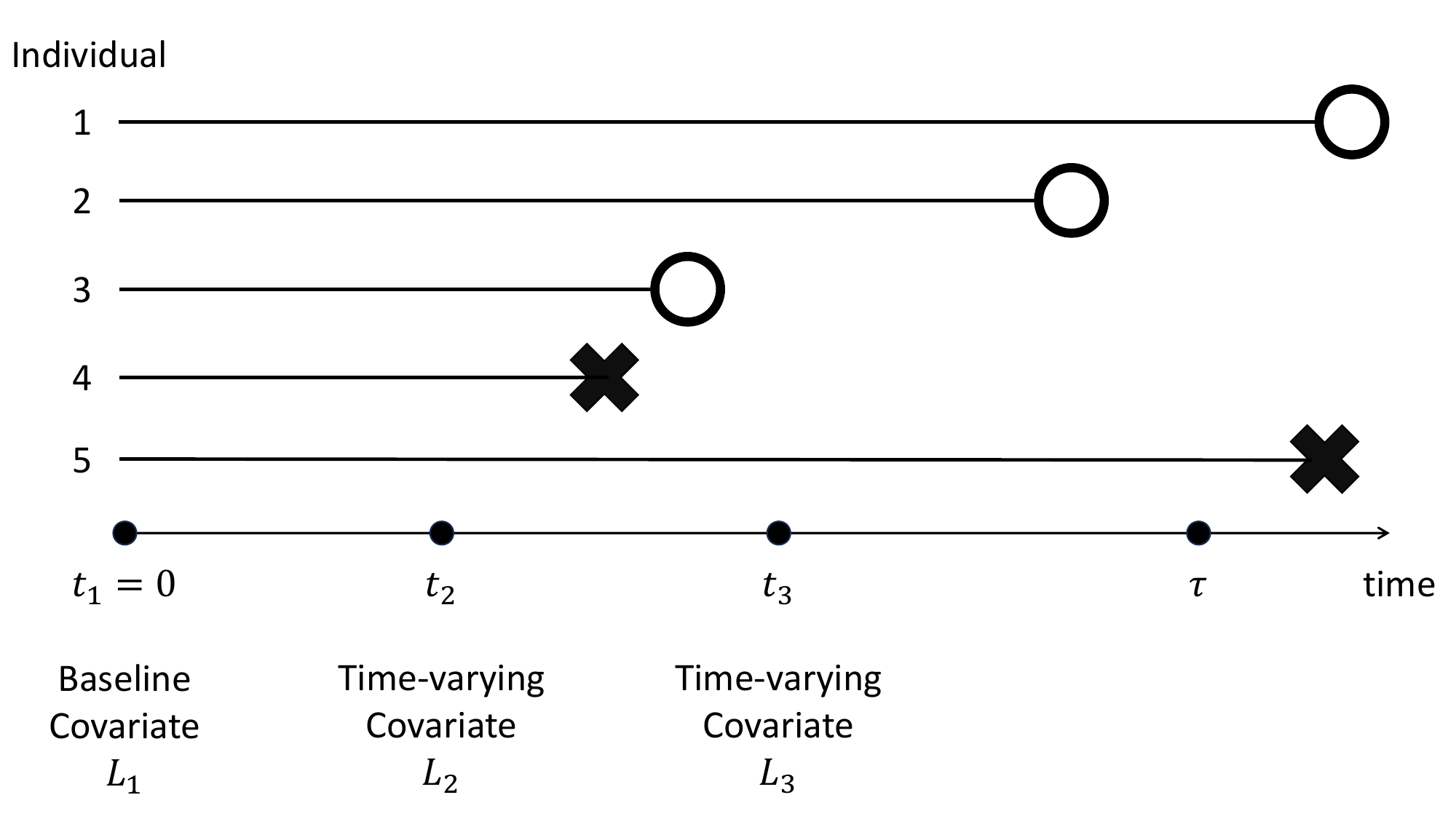}
    \caption{Illustration of the study design and data structure for the case $K=3$. Each horizontal line represents an individual's trajectory. An ending cross stands for an event of interest, such as disease or death; an ending circle stands for censoring, for example, due to loss of follow-up or end of the study. Individuals 3 and 4 experience the event or drop out of the study before $t_3$ and their time-varying covariates $L_3$ at $t_3$ are set to zero.}
    \label{fig: data structure}
\end{figure}

We next describe the estimand, the survival probability at a fixed given time point $\tau \in (0,\infty)$ conditional on historical covariates at a given visit time.
Without loss of generality, we assume that $\tau > t_K$ and define $t_{K+1} := \tau$.
Let $\varsigma \in \{1,2,\ldots,K\}$ be a given visit time index such that covariates of interest are measured at $t_\varsigma$.
Our first estimand is the function $\beta_{*}: \tilde{h}_\varsigma \mapsto \tilde{P}_*(T > \tau \mid T>t_\varsigma, \tilde{H}_\varsigma=\tilde{h}_\varsigma)$ for $\tilde{h}_\varsigma \in \mathcal{H}_\varsigma$.
This function describes the association between the historical covariate $\tilde{H}_\varsigma$ at time $t_\varsigma$ and the survival probability at a future time $\tau$, in the subpopulation that survives through time $t_\varsigma$ for which $\tilde{H}_\varsigma$ is observed under no censoring.
Our second estimand is the marginal survival probability $\alpha_* := \tilde{P}_*(T > \tau)$, which equals $\expect_{\tilde{P}_*}[\beta_*(\tilde{H}_\varsigma)]$ when $\varsigma=1$.
In a clinical trial setting, $\alpha_*$ for each treatment arm informs the treatment effect on the marginal survival.

We generally use a subscript $*$ to denote functions or quantities defined by the true distribution $\tilde{P}_*$ or $P_*$.
When writing expectations under $P_*$ or $\tilde{P}_*$, we may use the shorthand notation $\expect_*$.
We adopt the convention that $\sum_{i=a}^b \cdots = 0$ and $\prod_{i=a}^b \cdots = 1$ whenever $a > b$, and $0 \cdot c = 0$ regardless of whether $c$ is well defined and whether $c$ is finite.

We will use $a \lesssim b$ for two real numbers $a$ and $b$ to denote $a \leq \const b$ where $\const>0$ is a constant that may depend on $P_*$.
We sometimes use $\const$ to denote a generic positive constant.
For a function $f: \mathcal{H}_k \rightarrow \real$ ($k=1,\ldots,K$), we define its conditional $L^2(P_*)$-norm $\| f \|_{k} := \{ \expect_*[f(H_k)^2 \mid X>t_k] \}^{1/2}$. This definition is natural since $H_k$ is fully observed only conditional on $X>t_k$.
For a function $f: \mathcal{S} \times \mathcal{H}_k \rightarrow \real$ where $k=1,\ldots,K$ and $\mathcal{S} \subseteq (0,\infty)$, we define $\| f \|_{\mathcal{S},k} := \parallel \sup_{t \in \mathcal{S}} |f(t, \cdot)| \parallel_{k}$.
We will consider the asymptotic scenario where $n \to \infty$ with a fixed data-generating distribution $P_*$, and use big-O and little-o notations accordingly.

\section{Motivating application: a COVID-19 vaccine efficacy trial} \label{sec: example}

VAT00008 Stage 2 (NCT04904549) is a global phase~3 efficacy trial conducted in Colombia, Mexico, Kenya, Uganda, Ghana, India, and Nepal. The trial evaluated a Spike SARS-CoV-2 recombinant protein vaccine with AS03-adjuvant containing two Spike strains (Index/Ancestral and Beta/B.1.351) vs.\ placebo administered on Day~1 and Day~22 \protect\citep{Dayan2023}.
The trial enrolled 13,002 participants between October 19, 2021, and February 15, 2022, evaluating 11,416 adult participants for efficacy in the primary analysis cohort with 9,693 (84.9\%) baseline seropositive for SARS-CoV-2 (indicating previously infected) and the remainder baseline seronegative (naive to SARS-CoV-2).
The primary endpoint was virologically-confirmed SARS-CoV-2 infection with symptoms of COVID-19-like illness starting 14 days after the second injection, the so-called COVID-19 primary endpoint.
On Day~22 and Day~43, two biomarkers measuring SARS-CoV-2 antibody response to prior SARS-CoV-2 infection and/or vaccination were measured for each participant. These biomarkers are $\log 10$ concentration of IgG binding antibody against the Omicron-strain SARS-CoV-2 Spike protein and the $\log 10$ neutralizing antibody titer against the SARS-CoV-2 BA.1 strain.
Similar biomarkers have been studied in other COVID-19 vaccine biomarker correlates studies \protect\citep[e.g.,][]{gilbert2022covid, zhang2024omicron}.

There are two main analytical challenges: First, the follow-up of participants was right-censored by loss to follow-up and by outside vaccination with a non-study vaccine (typically mRNA vaccination), which boosts the antibody marker levels not due to the recombinant protein vaccination in the trial. The Day~43 antibody data can predict outside vaccination. Therefore, we need to account for potential informative censoring due to the Day~43 antibody level in addition to the covariates at the time origin (Day~22).
Second, the observed times to the COVID-19 endpoint or right-censoring were measured continuously in days. 
To fully leverage information in these continuously observed times, it is desirable to consider continuous time rather than discretize time.

We adopt the setup in Section~\ref{sec: setup} to frame the problem and these challenges.
We study each antibody biomarker separately.
The total number of visits $K=2$. The time origin $t_1=0$ is the Day~22 visit. The second visit time $t_2$ is the Day~43 visit. The baseline covariate $L_1$ consists of Day~22 post-treatment antibody level $W_a$, treatment arm $W_v$ (1=vaccine, 0=placebo), and baseline covariates $W_c$. The time-varying covariate $L_2$ at $t_2$ consists of Day~43 antibody level.
In the original analysis, COVID-19 endpoints within 6 days of the Day~43 visit may have been caused by SARS-CoV-2 infections before the Day~43 visit, which could perturb the Day~43 marker values, such that these endpoints were excluded from the analysis.
In this analysis, to leverage these endpoints while accounting for account for time-varying covariates, we set $t_2=43+6-22=27$ days to correspond to 6 days after the Day~43 visit and include these endpoints.
The visit time index $\varsigma$ of interest is 1, and the time point $\tau$ of interest is $t_2$ as well as 90, 135, 180, and 365 days after the Day~22 visit.

We next describe two sets of highly scientifically compelling estimands in terms of the conditional survival function $\beta_*$ or the marginal survival $\alpha_*$ described in Section~\ref{sec: setup}. Our first objective is to estimate the additive effect of the vaccine vs.\ placebo on the marginal cumulative incidence of the COVID-19 primary endpoint. With $\tilde{P}_{*,w_v}$ denoting the true distribution in treatment arm $w_v \in \{0,1\}$, we wish to estimate $\alpha_{*,w_v} := \tilde{P}_{*,w_v}(T > \tau)$, the marginal survival probability in treatment arm $w_v \in \{0,1\}$.
The additive effect equals $(1-\alpha_{*,1})-(1-\alpha_{*,0})$, a contrast of marginal survival probabilities $\alpha_{*,w_v}$.

Our second objective is to study the association between the two biomarkers and the cumulative incidence of the COVID-19 primary endpoint, particularly whether the vaccine's effect can be explained by its effect on Day~22 biomarkers.
We wish to estimate the conditional survival function $\beta_*: \ell_1=(w_a,w_v,w_c) \mapsto \tilde{P}_*(T > \tau \mid W_a=w_a,W_v=w_v,W_c=w_c)$, with $\varsigma=1$, to study the association between $W_a$ and $T$.
Next, to assess surrogacy, we define the controlled direct effect $\CDE(w_a,w_c)$ of the vaccine $W_v$ on the incident COVID-19 endpoint not mediated through the Day~22 antibody level $W_a$, conditional on baseline covariate $W_c=w_c$, with Day~22 antibody level $W_a$ set to $w_a$.
A Day~22 antibody marker is considered a strong surrogate endpoint if $\CDE(w_a,w_c) \approx 0$ for a wide plausible range of $(w_a,w_c)$, as $\CDE(w_a,w_c)$ near zero means that a large proportion of the total vaccine efficacy is eliminated by setting the antibody marker to the same value for both treatment arms.
Following Proposition 1 of \protect\citet{Vanderweele2011}, $\CDE(w_a,w_c)$ is causally identified under two primary assumptions: (i) no unmeasured confounding between the vaccine $W_v$ and the endpoint $T$, which is guaranteed by randomization; and (ii) no unmeasured confounding between the mediator $W_a$ and $T$, which holds if $(W_c, W_v)$ constitutes a sufficient confounding set.
Under these conditions, the nonparametric identification formula yields $\CDE(w_a,w_c)=\{1-\beta_{*}(w_a,w_v=1,w_c)\} - \{1-\beta_{*}(w_a,w_v=0,w_c)\}$, a contrast of the conditional survival function $\beta_*$.

\section{Identification formulas}
\label{sec: identification}

We assume the following \emph{coarsening at random} condition \protect\citep{Gill1997,VanderLaan2002} throughout, that is, future events are independent of the coarsening $(C \wedge t_{k+1})$ in the next time window given the observed history $(X>t_k, H_k)$.
This condition reduces to the following \emph{conditionally independent censoring} in our setting.

\begin{condition}[Conditionally independent censoring] \label{condition: cond ind cens}
    $(T \wedge \tau, \tilde{L}_{k+1},\ldots,\tilde{L}_K) \independent (C \wedge t_{k+1}) \mid X>t_k, \tilde{H}_k$ for all $k=1,\ldots,K$.
\end{condition}

This condition states that, given staying in the study up to visit time $t_k$, the historical covariates $\tilde{H}_k$ contain all information to explain any dependence between (i) the remaining time to event and all future covariates, and (ii) the time to censoring in the next window $(t_k,t_{k+1}]$. 

For $t \in (t_k,t_{k+1}]$, we define the following two survival functions pointwise:
\begin{align*}
    S_{*,k}(t \mid h_k) &:= \tilde{P}_*(T > t \mid X>t_k, H_k=h_k), \\
    G_{*,k}(t \mid h_k) &:= \tilde{P}_*(C > t \mid X>t_k, H_k=h_k).
\end{align*}

We also assume the following positivity condition throughout.
\begin{condition}[Positivity] \label{condition: positive prob uncensored}
    The probabilities $G_{*,k}(t_{k+1} \mid H_k)$ ($k=1,\ldots,K$) of not being censored are all positive $P_*$-almost surely.
\end{condition}

Condition~\ref{condition: positive prob uncensored} ensures that some observations are uncensored by the time of interest $\tau$, so that it is possible to learn about survival probability at $\tau$ without extrapolating.
Under the above conditions, $S_{*,k}$ and $G_{*,k}$ can both be nonparametrically identified within the time window $(t_k,t_{k+1}]$ using hazard functions conditional on $H_k$, after artificially right-censoring all observations with $X>t_k$ at $t_{k+1}$. Mechanically, this corresponds to artificially right-censoring all data at $t_{k+1}$ and applying a standard Kaplan-Meier estimator within each stratum defined by the historical covariates $H_k$ among at-risk individuals at $t_k$.
These conditions lead to the following Lemma~\ref{lemma: at-risk set} on the equivalence of conditioning on $X>t_k$ and $T>t_k$ given all historical covariates $\tilde{H}_k$, which subsequently leads to the identification formula in Theorem~\ref{thm: G identify}.

\begin{lemma} \label{lemma: at-risk set}
    Under Conditions~\ref{condition: cond ind cens} and \ref{condition: positive prob uncensored}, for any $k=1,\ldots,K$ and $h_{k-1} \in \mathcal{H}_{k-1}$,
    $$(\tilde{L}_k,\ldots,\tilde{L}_K,T \wedge \tau -t_k) \mid T>t_k,\tilde{H}_{k-1}=h_{k-1} \text{ and } (\tilde{L}_k,\ldots,\tilde{L}_K,T \wedge \tau -t_k) \mid X>t_k,H_{k-1}=h_{k-1}$$
    are identically distributed.
    Consequently, for any $t \in (t_k,t_{k+1}]$ and $h_{k} \in \mathcal{H}_{k}$,
    \begin{equation}
        S_{*,k}(t \mid h_k) = \tilde{P}_*(T>t \mid T>t_k, \tilde{H}_k=h_k), \quad G_{*,k}(t \mid h_k) = \tilde{P}_*(C > t \mid T>t_k, \tilde{H}_k=h_k). \label{eq: S G at-risk set}
    \end{equation}
\end{lemma}

\begin{theorem}[G-computation formula] \label{thm: G identify}
    Define recursively
    \begin{align*}
        & U_{*,K}(H_K) := 1, \qquad Q_{*,k}(H_k) := U_{*,k}(H_k) S_{*,k}(t_{k+1} \mid H_k) \quad (k=K,\ldots,1), \\
        & U_{*,k}(H_k) := \expect_*[Q_{*,k+1}(H_{k+1}) \mid X>t_{k+1}, H_k] \quad (k=K-1,\ldots,1).
    \end{align*}
    Under Conditions~\ref{condition: cond ind cens} and \ref{condition: positive prob uncensored}, $Q_{*,k}(H_k) = \tilde{P}_*(T>\tau \mid T>t_k,H_k)$ and $U_{*,k} = \tilde{P}_*(T>\tau \mid T>t_{k+1},H_k)$ ($k<K$) $P_*$-almost surely. Thus, taking $k=\varsigma$, $\beta_{*}(H_\varsigma) = Q_{*,\varsigma}(H_\varsigma)$ $P_*$-almost surely.
\end{theorem}

Both $Q_{*,k}$ and $U_{*,k}$ are conditional survival probabilities at $\tau$ given history $H_k$, with the difference that $Q_{*,k}$ and $U_{*,k}$ condition on survival beyond $t_k$ and $t_{k+1}$, respectively.
As such, they are connected through a product with the survival $S_{*,k}$ beyond $t_{k+1}$ given survival beyond $t_k$ in the definition, in a fashion similar to hazards.
Since the same identification holds for any $\varsigma=1,\ldots,K$, we may assume $\varsigma=1$ without loss of generality when studying the survival probability $\beta_*$ conditional on the entire history $\tilde{H}_\varsigma$.

\begin{remark}
    We allow $U_{*,k}(h_k)$ to be undefined for $h_k \in \mathcal{H}_k$ such that $P_*(X>t_{k+1} \mid X>t_k,H_k=h_k)=0$.
    Under Conditions~\ref{condition: cond ind cens} and \ref{condition: positive prob uncensored}, such $h_k$ must yield $S_{*,k}(t_{k+1} \mid h_k) = 0$ and so $Q_{*,k}$ is still well defined in the entire support $\mathcal{H}_k$ under the convention that $0 \cdot c=0$ if $c$ is undefined.
\end{remark}

Moreover, a SDR identification formula can be derived. For any conditional survival functions $S_k, G_k: (t_k, t_{k+1}] \times \mathcal{H}_k \to [0,1]$, $x,t \in (t_k,t_{k+1}]$ and $\delta \in \{0,1\}$, define pointwise
\begin{align*}
    &\TS_{k,t}(S_k,G_k)(x,\delta \mid h_k) \\
    &:= S_k(t \mid h_k) -S_k(t \mid h_k) \Bigg\{ \frac{\ind(x \leq t) \delta}{S_k(x \mid h_k) G_k(x- \mid h_k)} + \int_{(t_k,x \wedge t]} \frac{S_k(\intd s \mid h_k)}{S_k(s \mid h_k) S_k(s- \mid h_k) G_k(s- \mid h_k)} \Bigg\}
\end{align*}
whenever $G_k(t- \mid h_k) > 0$ and zero otherwise.
Here, the integral is a Lebesgue-Stieltjes integral with respect to $S_k(\cdot \mid h_k)$.
This is a one-step corrected transformation of $S_k$, which resembles the doubly robust transformation for estimating conditional mean counterfactual and conditional average treatment effect \protect\citep[Section~3.1 in][]{VanDerLaan2013}, and leverages the influence function of marginal survival probabilities \protect\citep{Bai2013,Hubbard2000}.
In this transformation, the term $\ind(x \leq t) \delta/\{S_k(x \mid h_k) G_k(x- \mid h_k)\}$ is an inverse probability-of-censoring weighting to correct bias due to censoring; the integral is an adjustment based on the martingale to achieve nonparametric efficiency and double robustness; the multiplication with $S_k(t \mid h_k)$ arises from applying Delta method to restore the survival probability scale.
This transformation is doubly robust in the sense that, under Conditions~\ref{condition: cond ind cens} and \ref{condition: positive prob uncensored}, for any conditional survival functions $S_k$ and $G_k$,
\begin{align}
    & S_{*,k}(t \mid H_k) = \expect_* \left[ \TS_{k,t}(S_k,G_k)(X,\Delta \mid H_k) \mid X>t_k, H_k \right] \text{ $P_*$-almost surely} \label{eq: DR transform} \\
    & \text{for all $t \in (t_k,t_{k+1}]$ if (i) $S_k=S_{*,k}$, or (ii) $G_k=G_{*,k}$.} \nonumber
\end{align}
A proof of this equation can be found in Supplement~\ref{sec: proof identify}.
Moreover, define pointwise
\begin{align}
    &\T_{k}( \{U_j, S_j, G_j\}_{j=k}^K)(O) \nonumber \\
    &:= \sum_{j=k+1}^{K} \ind(X>t_j) \left\{ \prod_{\ell=k}^{j-1} \frac{1}{G_{\ell}(t_{\ell+1} \mid H_{\ell})} \right\} \{ U_j(H_j) \TS_{j,t_{j+1}}(S_j,G_j)(X,\Delta \mid H_j) - U_{j-1}(H_{j-1}) \} \nonumber \\
    &\qquad+ U_k(H_k) \TS_{k,t_{k+1}}(S_k,G_k)(X,\Delta \mid H_k) \label{eq: SDR transform}
\end{align}
whenever $G_k(t_{k+1} \mid H_k) > 0$ for all $k$ and zero otherwise, where $U_K:=1$.
We derive $\T_{k}$ by induction backward in time stage $k$, and show that $\T_{k}( \{U_j, S_j, G_j\}_{j=k}^K) - Q_{*,k}$ is the conditional influence function \protect\citep{Chernozhukov2024} for estimating $Q_{*,k}$ under a nonparametric model.
Basically, this function characterizes the first-order change in $Q_{*,k}(h_k)$ under a smooth local perturbation of $P_*$.
A similar result about the marginal survival $\alpha_*$ in terms of influence functions is in Theorem~\ref{thm: efficiency} in Section~\ref{sec: marg SDR estimator}, and we prove the claim that $\T_{k}( \{U_j, S_j, G_j\}_{j=k}^K) - Q_{*,k}$ is the conditional influence function when proving Theorem~\ref{thm: efficiency}.
This conditional influence function has a similar stage-wise structure as in the discrete-time case \protect\citep{Rotnitzky2017,Luedtke2017}, where $U_j(H_j) \TS_{j,t_{j+1}}(S_j,G_j)(X,\Delta \mid H_j)$ serves as a stage-wise ``outcome'', $U_{j-1}(H_{j-1})$ or $Q_{*,k}(H_\varsigma)$ serves as a stage-wise ``regression model'', and $\prod_{\ell=k}^{j-1} G_{\ell}(t_{\ell+1} \mid H_{\ell})^{-1}$ serves as a stage-wise inverse probability weight. The presence of the doubly robust transformation of nuisance survival functions $(S_j,G_j)$ that allows for continuous time, namely $\TS_{j,t_{j+1}}(S_j,G_j)(X,\Delta \mid H_j)$, along with the stage-wise regression $U_j(H_j)$ and the stage-wise weight $\prod_{\ell=k}^{j-1} G_{\ell}(t_{\ell+1} \mid H_{\ell})^{-1}$, accommodates both discrete visit times and continuous time to event or censoring.

Eq.~\ref{eq: SDR transform} is a one-step corrected transformation of $Q_{*,k}$ based on this conditional influence function. This correction leverages a first-order approximation to the functional $P_* \mapsto Q_{*,k}$ based on the conditional influence function, and can thus remove the first-order impact of estimation errors in the nuisance functions $\{U_j, S_j, G_j\}_{j=k}^K$ on estimators of this functional.
In the following Theorem~\ref{thm: SDR identify}, we highlight the SDR property \protect\citep{Rotnitzky2017,Luedtke2017} of \eqref{eq: SDR transform}, a particularly useful property in estimation.
\begin{theorem}[SDR formula] \label{thm: SDR identify}
    Under Conditions~\ref{condition: cond ind cens} and \ref{condition: positive prob uncensored}, for all $k=1,\ldots,K$,
    $$\expect_*[\T_{k}( \{U_j, S_j, G_j\}_{j=k}^K)(O) \mid X>t_k, H_k] = Q_{*,k}(H_k) \text{ $P_*$-almost surely}$$
    if, for all $j = k,\ldots,K$, (i) $S_j=S_{*,j}$ for all $t \in (t_j,t_{j+1}]$ and $U_j=U_{*,j}$, or (ii) $G_j=G_{*,j}$ for all $t \in (t_j,t_{j+1}]$.
\end{theorem}

Since Theorem~\ref{thm: G identify} establishes $Q_{*,\varsigma}=\beta_{*}$, Theorem~\ref{thm: SDR identify} yields a SDR identification formula for $\beta_{*}$ by setting $k=\varsigma$.
Because our problem is featured in a mixture of continuous time and covariates collected on a discrete time grid, both $S_{*,j}$ and $U_{*,j}$ contain some information of the outcome.
Hence, Theorem~\ref{thm: SDR identify} shows a similar SDR structure---with an additional outcome model---to discrete-time problems \protect\citep{Rotnitzky2017,Luedtke2017}: the conditional mean is Fisher consistent if, in each time window $(t_j,t_{j+1}]$, either the pair of outcome models ($S_j$ and $U_{j}$) or the censoring distribution ($G_j$) is correctly specified.
Theorem~\ref{thm: SDR identify} goes beyond existing results by applying in continuous time with a novel combination of these nuisance functions. 

An inverse probability of censoring weighting (IPCW) formula also follows from Theorem~\ref{thm: SDR identify} by setting $S_j=1$, $G_j=G_{*,j}$ ($j=1,\ldots,K$), and $U_j=0$ ($j=1,\ldots,K-1$).
Though this approach is not SDR, it may lead to convenient regression estimators, similarly to some marginal structural model estimators \protect\citep[e.g.,][]{Hernan2000}.
Another IPCW formula can be derived by setting $U_j=1$ ($j=1,\ldots,K-1$) instead, but this formula may be tedious to compute because it represents $\beta_*$ as a conditional expectation of weighted sums of events over all time windows $(t_k,t_{k+1}]$ ($k=\varsigma,\ldots,K$).

\begin{corollary}[IPCW formula] \label{coro: IPCW identify}
    Under Conditions~\ref{condition: cond ind cens} and \ref{condition: positive prob uncensored},
    $$\beta_{*}(H_\varsigma) = \expect_* \left[ \ind(X>t_K) \left\{ \prod_{j=\varsigma}^{K-1} \frac{1}{G_{*,j}(t_{j+1} \mid H_j)} \right\} \left\{ 1 - \frac{\ind(X \leq \tau) \Delta}{G_{*,K}(X- \mid H_K)} \right\} \mid X>t_\varsigma, H_\varsigma \right]$$
    $P_*$-almost surely.
\end{corollary}

In Section~\ref{sec: SDR estimator}, we first present the estimator of the conditional survival probability $\beta_{*}$ that leverages the SDR transformation \eqref{eq: SDR transform} and is justified based on Theorem~\ref{thm: SDR identify}.
In Section~\ref{sec: marg SDR estimator}, we present the SDR estimator of the marginal survival probability $\alpha_*$, which builds on the SDR conditional survival estimator.

\section{Sequentially doubly robust estimation of conditional survival probability function} \label{sec: SDR estimator}

We propose to first estimate nuisance functions $\{S_{*,k},G_{*,k}\}_{k=\varsigma}^K$ flexibly with $\{\hat{S}_k,\hat{G}_k\}_{k=\varsigma}^K$, then compute a pseudo-outcome based on the SDR transformation in \eqref{eq: SDR transform}, and regress the pseudo-outcome onto the history iteratively from $k=K$ to $k=\varsigma$.
The SDR property of our estimator hinges on regressing the SDR transformed pseudo-outcome, rather than a na\"ive product of $\hat{S}_k$ and $\hat{U}^\SDR_{k+1}$ motivated by the identification in Theorem~\ref{thm: G identify}, when estimating $U_{*,k}$ as well as $Q_{*,\varsigma}$.
We go beyond the discrete-time SDR estimators \protect\citep[e.g.,][]{Rotnitzky2017,Luedtke2017} by applying our novel continuous-time SDR transformation.
Roughly speaking, the SDR estimator is constructed as follows: we first estimate nuisance conditional survival functions $(S_{*,k},G_{*,k})$, and iterate through time stages $k=K$ to $k=\varsigma$; at each stage, we first compute a real-valued pseudo-outcome based on \eqref{eq: SDR transform} and the nuisance estimators, and regress the pseudo-outcome on $H_{k-1}$ in the sample to estimate $U_{*,k-1}$ for $k>\varsigma$, or on $H_k$ to estimate $Q_{*,k}$ for $k=\varsigma$.
For illustration, we describe the algorithm without sample splitting in Algorithm~\ref{alg: SDR}.
The regression in Steps~\ref{alg step: regress U} and \ref{alg step: regress} can be implemented with any supervised learning algorithm, such as linear regression and gradient boosting.

\begin{algorithm}[tb]
    \caption{SDR estimator of $\beta_{*}$} \label{alg: SDR}
    \begin{algorithmic}[1]
        \Require Data $\{O_1,\ldots,O_n\}$, algorithm to estimate nuisance conditional survival functions $S_{*,k}$ and $G_{*,k}$ ($k=\varsigma,\ldots,K$), regression algorithm to estimate conditional mean.
        \State Estimate nuisance conditional survival functions: For all $k=\varsigma,\ldots,K$, estimate $(t \mid h_k) \mapsto S_{*,k}(t \mid h_k)$ and $(t \mid h_k) \mapsto G_{*,k}(t \mid h_k)$ with $\hat{S}_k$ and $\hat{G}_k$, respectively, for $t \in (t_k,t_{k+1}]$ and $k=\varsigma,\ldots,K$. The function estimators $t \mapsto \hat{S}_k(t \mid h_k)$ and $t \mapsto \hat{G}_k(t \mid h_k)$ must be non-negative decreasing c\`{a}dl\`{a}g for all $h_k$ and satisfy $\hat{S}_k(t_k \mid h_k) = \hat{G}_k(t_k \mid h_k) = 1$. \label{alg step: est S G}
        \State Set $\hat{U}_{K}^\SDR=1$.
        \For{$k=K,\ldots,\varsigma$}
        \State Compute the pseudo-outcome $\T_{k}( \{\hat{U}_{j}^\SDR, \hat{S}_j, \hat{G}_j\}_{j=k}^K)(O)$
        defined in \eqref{eq: SDR transform} for each individual at risk at visit time $t_k$ ($X>t_k$). \label{alg step: pseudo-outcome}
        \If{$k>\varsigma$}
            \State Regress $\T_{k}( \{\hat{U}_{j}^\SDR, \hat{S}_j, \hat{G}_j\}_{j=k}^K)(O)$ on $H_{k-1}$ among individuals with $X>t_k$. The fitted predictive model $\hat{U}_{k-1}^\SDR$ is the estimator of $U_{*,k-1}$. \label{alg step: regress U}
        \Else
            \State Regress $\T_{k}( \{\hat{U}_{j}^\SDR, \hat{S}_j, \hat{G}_j\}_{j=k}^K)(O)$ on $H_k$ among individuals with $X>t_k$. The fitted predictive model $\hat{Q}_{k}^\SDR$ is the estimator of $Q_{*,k}$. \label{alg step: regress}
        \EndIf
        \EndFor
        \State\Return $\hat{Q}_{\varsigma}^\SDR$. \label{alg step: output}
    \end{algorithmic}
\end{algorithm}

In practice, for each $k$, we propose to estimate $(S_{*,k},G_{*,k})$ by artificially right-censoring individuals with $X>t_k$ (at-risk at $t_k$) at $t_{k+1}$ and applying conditional survival function estimators that assume conditional independent censoring (Condition~\ref{condition: cond ind cens}) to this artificially right-censored subset. The additional right-censoring step avoids bias from using future observations after $t_{k+1}$, because these observations require future covariates to account for future censoring. If Condition~\ref{condition: cond ind cens} is assumed to hold for $\tau = \infty$, then individuals with $X>t_K$ need not be artificially right-censored.
The algorithms to obtain estimators $(\hat{S}_k,\hat{G}_k)$ include the Cox proportional hazards model \protect\citep{Andersen1982,Cox1972} in conjunction with the Breslow estimator \protect\citep{Breslow1972}, the accelerated failure time model \protect\citep{Wei1992}, the piecewise
constant hazard model \protect\citep{Friedman1982}, Kaplan-Meier estimator in conjunction with nearest neighbor \protect\citep{Beran1981,Chen2019} or kernel methods \protect\citep{Beran1981,Chen2019,Dabrowska1989}, survival trees \protect\citep{LeBlanc1993,Segal1988}, random survival forests \protect\citep{Ishwaran2008}, the Cox model in conjunction with neural networks \protect\citep{Katzman2018,Kvamme2019}, global or local survival stacking \protect\citep{Craig2021,Wolock2024}, or an ensemble of several aforementioned methods via survival Super Learner \protect\citep{Westling2023}.

Sample splitting or cross-fitting \protect\citep{Schick1986,Klaassen2007,Kennedy2022}, namely training and evaluating the nuisance functions (Lines~\ref{alg step: est S G} and \ref{alg step: pseudo-outcome}, respectively) on independent samples, can also be applied in Algorithm~\ref{alg: SDR} to allow more flexibility in the nuisance function estimators $\{\hat{S}_k,\hat{G}_k,\hat{U}_k^\SDR\}$ and to improve performance in small to moderate samples.
We discuss cross-fitting in more details and describe our proposed cross-fit algorithm in Algorithm~\ref{alg: CV SDR} in Supplement~\ref{sec: CV}.
G-computation and IPCW estimators can be constructed similarly based on Theorem~\ref{thm: G identify} and Corollary~\ref{coro: IPCW identify}, respectively.
We use superscripts $\G$ and $\IPCW$ instead of $\SDR$ to denote these estimators, and describe them in Algorithms~\ref{alg: Gcomp} and \ref{alg: IPCW} in Supplement~\ref{sec: G IPCW alg}, respectively.

The estimator from Algorithm~\ref{alg: SDR} may yield values outside the unit interval $[0,1]$ even though the estimand is a probability, if Step~\ref{alg step: regress} does not respect this constraint, for example, when using linear regression.
Moreover, when multiple time points $\tau$ are of interest, the estimated conditional survival probability might fail to be monotonically non-increasing in $\tau$, even though the estimand is monotone in $\tau$. To alleviate such incompatibility with the estimand's natural constraints, we propose to post-process the estimators: We first project the raw estimator onto the unit interval $[0,1]$. When estimating $\beta_*$ for finitely many $\tau$, we propose to run an isotonic regression of the estimates on $\tau$ to ensure monotonicity \protect\citep{Westling2020}.
This post-processing step ensures compatibility with the parameter space while preserving the first-order asymptotic properties of $\hat{Q}_{\varsigma}^\SDR$.

We next present theoretical results about $\hat{Q}_{\varsigma}^\SDR$.
We define the following oracle estimators
\begin{align}
    &\bar{U}_{K}^\SDR := 1, \nonumber \\
    &\bar{U}_{k}^\SDR: h_k \mapsto \expect_*[\T_{k+1}( \{\hat{U}_{j}^\SDR, \hat{S}_j, \hat{G}_j\}_{j=k+1}^K)(O) \mid X>t_{k+1}, H_k=h_k], \quad (k=K-1,\ldots,\varsigma) \label{eq: oracle U} \\
    &\bar{Q}_{\varsigma}^\SDR: h_\varsigma \mapsto \expect_*[\T_{k}( \{\hat{U}_{j}^\SDR, \hat{S}_j, \hat{G}_j\}_{j=\varsigma}^K)(O) \mid X>t_\varsigma, H_\varsigma=h_\varsigma] \label{eq: oracle Q}
\end{align}
of $U_{*,k}$ and $Q_{*,\varsigma}$, respectively, where $O$ denotes a random draw from $P_*$ independent of the data.
Throughout this paper, when random functions such as $\hat{U}_{k}^\SDR$, $\hat{S}_{k}$ and $\hat{G}_{k}$ appear in an expectation, we treat these functions as fixed when taking the expectation. For example, $\bar{U}_{k-1}^\SDR$ is random through $\{\hat{U}_{j}^\SDR, \hat{S}_j, \hat{G}_j\}_{j=k}^K$ because the expectation does not account for the randomness in these estimators.
These oracle estimators are the \emph{population-level} regression of the pseudo-outcomes with \emph{nuisance estimators}.
Since $U_{*,k}$ is the population-level regression of pseudo-outcomes with \emph{true nuisances}, and $\hat{U}_{k}^\SDR$ is the \emph{sample-level} regression of pseudo-outcomes with nuisance estimators, we analyze $\hat{U}_{k}^\SDR-U_{*,k}$ via the decomposition
$$\hat{U}_{k}^\SDR-U_{*,k} = \underbrace{(\bar{U}_{k}^\SDR - U_{*,k})}_{\text{bias due to nuisance estimation}} + \underbrace{(\hat{U}_{k}^\SDR-\bar{U}_{k}^\SDR)}_{\text{regression estimation error}}.$$
We study $\hat{Q}_{\varsigma}^\SDR-Q_{*,\varsigma}$ similarly.
By studying this decomposition, we treat nuisance and regression estimators as black boxes and provide high-level theoretical results about $\hat{Q}_\varsigma^\SDR$, which can be combined with theoretical results for any adopted nuisance and regression estimators to provide concrete theoretical results for $\hat{Q}_\varsigma^\SDR$.
Similar decompositions have been adopted elsewhere---for example, Theorem~2 in \protect\citet{Yang2023} bounds a similar bias term; when estimating a marginal parameter, such a bias term corresponds to the higher-order remainder for a large class of asymptotically efficient estimators \citep[e.g.,][]{Kennedy2022}.
The oracle estimators for G-computation and IPCW estimators are defined similarly in \eqref{eq: G IPCW oracle} in Supplement~\ref{sec: G IPCW alg}, with the superscript $\SDR$ replaced by $\G$ and $\IPCW$ respectively.
We first present results about the bias.

\begin{condition}[Strong positivity] \label{condition: consistency of strong positivity}
    For all $k \geq \varsigma$, $G_{*,k}(t_{k+1} \mid H_k)$ is bounded away from zero $P_*$-almost surely; with probability tending to one, $\hat{G}_{k}(t_{k+1} \mid h_k)$ is also bounded away from zero for $P_*$-almost every $h_k$.
\end{condition}

The first half of this condition is a slightly stronger version of Condition~\ref{condition: positive prob uncensored}, whereas the second half imposes a similar strong positivity on the nuisance estimator $\hat{G}_k$.
In the following Theorem~\ref{thm: bound}, we define $\| 0 \|_{K+1} := 0$.

\begin{theorem}[Bounds on oracle estimators' biases] \label{thm: bound}
    \leavevmode
    \begin{itemize}
        \item SDR oracle bias: Suppose that Conditions~\ref{condition: cond ind cens} and \ref{condition: consistency of strong positivity} hold, and that $\hat{U}_{k}^\SDR$ is $P_*$-almost surely bounded by some constant with probability tending to one.
        For all $k=\varsigma,\ldots,K$, any square-integrable function $U_j$ where $U_K=1$, any conditional survival functions $S_j$ and $G_j$ ($j=\varsigma,\ldots,K$), define
        \begin{align*}
            \mathscr{B}_{k}(\{U_j, S_j, G_j\}_{j=k}^K) &:= \sum_{j=k}^K \Big\{ \| S_{j} - S_{*,j} \|_{(t_j,t_{j+1}],j} \wedge \| G_{j} - G_{*,j} \|_{(t_j,t_{j+1}],j} \\
            &\qquad+ \| G_{j} - G_{*,j} \|_{(t_j,t_{j+1}],j} \| U_j - U_{*,j} \|_{j+1} \Big\}.
        \end{align*}
        With $\bar{U}_{k}^\SDR$ in \eqref{eq: oracle U} and $\bar{Q}_{\varsigma}^\SDR$ in \eqref{eq: oracle Q}, with probability tending to one,
        \begin{align}
            &\left\| \bar{U}_{k}^\SDR-U_{*,k} \right\|_{k+1} \lesssim \mathscr{B}_{k+1}(\{\hat{U}_{j}^\SDR, \hat{S}_j, \hat{G}_j\}_{j=k+1}^K), \qquad (k=K-1,\ldots,\varsigma) \nonumber \\
            &\left\| \bar{Q}_{\varsigma}^\SDR-Q_{*,\varsigma} \right\|_\varsigma \lesssim \mathscr{B}_{\varsigma}(\{\hat{U}_{j}^\SDR, \hat{S}_j, \hat{G}_j\}_{j=\varsigma}^K). \label{eq: SDR point bound}
        \end{align}

        \item G-computation oracle bias: Under Conditions~\ref{condition: cond ind cens} and \ref{condition: consistency of strong positivity}, 
        if $\hat{U}_{k}^\G$ is $P_*$-almost surely bounded by some constant with probability tending to one, then for all $k=\varsigma,\ldots,K-1$, with probability tending to one,
        \begin{align}
            &\| \bar{U}_{k}^\G - U_{*,k} \|_{k+1} \lesssim \| \hat{S}_{k+1} - S_{*,k+1} \|_{(t_{k+1},t_{k+2}],k+1} + \|\hat{U}_{k+1}^\G - U_{*,k+1} \|_{k+2}, \nonumber \\
            &\| \bar{Q}_{\varsigma}^\G - Q_{*,\varsigma} \|_\varsigma \lesssim \| \hat{S}_\varsigma - S_{*,\varsigma} \|_{(t_\varsigma,t_{\varsigma+1}],\varsigma} + \| \hat{U}_{\varsigma}^\G - U_{*,\varsigma} \|_{\varsigma+1}. \label{eq: G point bound}
        \end{align}

        \item IPCW oracle bias: Under Conditions~\ref{condition: cond ind cens} and \ref{condition: consistency of strong positivity}, with probability tending to one,
        \begin{equation}
            \left\| \bar{Q}_{\varsigma}^\IPCW-Q_{*,\varsigma} \right\|_\varsigma \lesssim \sum_{j=\varsigma}^K \| \hat{G}_{j} - G_{*,j} \|_{(t_j,t_{j+1}],j}. \label{eq: IPCW point bound}
        \end{equation}
    \end{itemize}

\end{theorem}

In this theorem, the boundedness condition on the estimators of $U_{*,k}$ is mild because $U_{*,k}$ is bounded by one.

Apart from the biases in Theorem~\ref{thm: bound}, the estimation errors $\| \hat{U}_{k}^\SDR - \bar{U}_{k}^\SDR \|_{k+1}$ and $\| \hat{Q}_{\varsigma}^\SDR - \bar{Q}_{\varsigma}^\SDR \|_\varsigma$ of the oracle estimators $\bar{U}_{k}^\SDR$ and $\bar{Q}_{\varsigma}^\SDR$ also need to be studied.
These errors can often be shown to be $\smallo_p(1)$ by arguments showing consistency of regression estimators, for example, if the regression estimator is an empirical risk minimizer in a function class with restricted flexibility (i.e., $P_*$-Donsker function class) that can well approximate the truth \protect\citep[see, e.g., Chapter~3.2 in][]{vandervaart1996}.
This condition of vanishing estimation errors can be viewed as a generalization of correct model specification in parametric models to nonparametric regression, or a generalization of consistency to random estimands.
Thus, we anticipate these estimation errors to vanish if the regression technique is sufficiently flexible to approximate the oracle estimators well.
We slightly abuse the concept of consistency and say that a random function $\hat{f}$ is consistent for another random function $\bar{f}$ if $\| \hat{f} - \bar{f} \| = \smallo_p(1)$ for an appropriate norm.
Theorem~\ref{thm: bound} implies the following SDR property of $\hat{Q}_{\varsigma}^\SDR$ in terms of the above generalization of consistency or correct specification.

\begin{corollary}[Sequential double robustness] \label{coro: SDR}
    Suppose that $\hat{U}_k^\SDR$ is $P_*$-almost surely bounded by some constant with probability tending to one, and that Conditions~\ref{condition: cond ind cens} and \ref{condition: consistency of strong positivity} hold. If (i) $\| \hat{Q}_{\varsigma}^\SDR-\bar{Q}_{\varsigma}^\SDR \|_\varsigma = \smallo_p(1)$, and (ii) for each $j = \varsigma,\ldots,K$, $(\| \hat{S}_j - S_{*,j} \|_{(t_j,t_{j+1}],j}, \| \hat{U}_j^\SDR - \bar{U}_j^\SDR \|_{j+1}) = \smallo_p(1)$ or $\| \hat{G}_j - G_{*,j} \|_{(t_j,t_{j+1}],j} = \smallo_p(1)$, then $\| \hat{Q}_{\varsigma}^\SDR-Q_{*,\varsigma} \|_\varsigma = \smallo_p(1)$.
\end{corollary}

In other words, $\hat{Q}_{\varsigma}^\SDR$ is consistent for the estimand $Q_{*,\varsigma}$ if it is consistent for its oracle estimator $\bar{Q}_{\varsigma}^\SDR$ and, for each $j=\varsigma,K$, either $\hat{G}_j$ is consistent, or $\hat{S}_j$ is consistent and $\hat{U}_j^\SDR$ is consistent for its oracle estimator $\bar{U}_j^\SDR$.
In contrast, in general, the G-computation estimator $\hat{Q}_{\varsigma}^\G$ is not guaranteed to be consistent if $\hat{S}_k$ is inconsistent or $\| \hat{U}_{k}^\G-\bar{U}_{k}^\G \|_{k+1} \neq \smallo_p(1)$ for some $k \geq \varsigma$; the IPCW estimator $\hat{Q}_{\varsigma}^\IPCW$ is not guaranteed to be consistent if $\hat{G}_k$ is inconsistent for some $k \geq \varsigma$.
Thus, these two estimators are not sequentially doubly robust in general.

\begin{remark} \label{remark: product remainder}
    For many (sequentially) doubly robust estimators, due to the product structure, namely mixed bias property or product bias property \protect\citep{Rotnitzky2021}, of the higher-order remainder, the convergence rate can be bounded by a product of convergence rates of several nuisance functions using Cauchy-Schwarz inequality \protect\citep[e.g.,][]{Luedtke2017,VanderLaan2018}.
    In contrast, the rate we obtain in Theorem~\ref{thm: bound} for the SDR estimator contains terms that are essentially the minimum of two nuisance function estimators' convergence rates.
    This appears to be a particular challenge in survival settings, because Lebesgue-Stieltjes integrals are involved in \eqref{eq: SDR transform} and the higher-order remainder is no longer a simple product but a so-called \emph{cross integrated error term} \protect\citep{Ying2023}.
    In the marginal case, this term is shown in \eqref{eq: DR remainder} in Section~\ref{sec: marg SDR estimator}.
    This term is an integral of cross errors over time, and so cannot be bounded by Cauchy-Schwarz as for the product in the discrete-time case.
    In Supplement~\ref{sec: product remainder}, we establish that the remainder can be bounded by a product form when both $T$ and $C$ are discrete, that is, the terms $\| S_{j} - S_{*,j} \|_{(t_j,t_{j+1}],j} \wedge \| G_{j} - G_{*,j} \|_{(t_j,t_{j+1}],j}$ in \eqref{eq: SDR point bound} can be strengthened to a product form $\| S_{j} - S_{*,j} \|_{(t_j,t_{j+1}],j} \| G_{j} - G_{*,j} \|_{(t_j,t_{j+1}],j}$.
    Thus, we conjecture that the bound we obtained in Theorem~\ref{thm: bound} for the SDR estimator can also be improved for general distributions on time such as continuous time.
\end{remark}

\section{Sequentially doubly robust estimator of marginal survival probability} \label{sec: marg SDR estimator}

Algorithms~\ref{alg: SDR}, \ref{alg: Gcomp} and \ref{alg: IPCW} can be slightly modified to estimate the marginal survival probability in the entire population, namely $\alpha_*=\tilde{P}_*(T > \tau)$. We set $\varsigma=1$; when estimating $Q_{*,1}$ by regression for the first time point $t_1$, we regress the pseudo-outcome on the intercept rather than $H_1$, that is, take the average of the pseudo-outcome.
This approach is valid because $\alpha_* = \expect_*[\beta_*(L_1)]$.
The simplified algorithm for this SDR estimator is in Algorithm~\ref{alg: SDR marg}.

\begin{algorithm}
    \caption{SDR estimator of marginal survival probability $\alpha_*$} \label{alg: SDR marg}
    \begin{algorithmic}[1]
        \Require Data $\{O_1,\ldots,O_n\}$, algorithm to estimate $S_{*,k}$ and $G_{*,k}$ ($k=\varsigma,\ldots,K$), regression algorithm to estimate conditional mean. 
        \State Estimate $S_{*,k}$ and $G_{*,k}$ as in Line~\ref{alg step: est S G} of Algorithm~\ref{alg: SDR}.
        \State Set $\hat{U}_{K}^\SDR=1$.
        \For{$k=K,\ldots,1$}
        \State Compute the pseudo-outcome $\T_{k}( \{\hat{U}_{j}^\SDR, \hat{S}_j, \hat{G}_j\}_{j=k}^K)(O)$
        defined in \eqref{eq: SDR transform} for each individual at risk at visit time $t_k$ ($X>t_k$).
        \If{$k>1$}
            \State Regress $\T_{k}( \{\hat{U}_{j}^\SDR, \hat{S}_j, \hat{G}_j\}_{j=k}^K)(O)$ on $H_{k-1}$ among individuals with $X>t_k$. The fitted predictive model $\hat{U}_{k-1}^\SDR$ is the estimator of $U_{*,k-1}$.
        \Else
            \State Set $\hat{Q}_{0}^\SDR = n^{-1} \sum_{i=1}^n \T_{1}( \{\hat{U}_{j}^\SDR, \hat{S}_j, \hat{G}_j\}_{j=1}^K)(O_i)$. \label{alg step: average}
        \EndIf
        \EndFor
        \State\Return $\hat{Q}_{0}^\SDR$.
    \end{algorithmic}
\end{algorithm}

\begin{remark} \label{rmk: marginal for varsigma>1}
    It may be of interest to estimate $\alpha_{*,\varsigma} := \tilde{P}_*(T>\tau \mid T>t_\varsigma)$ marginalized over historical covariates $\tilde{H}_\varsigma$ for $\varsigma>1$. However, restricting our marginal estimator to $\varsigma = 1$ significantly simplifies identification and estimation. When $\varsigma=1$ so that $t_\varsigma=0$, the event being conditioned on, $T>t_\varsigma$ or $X>t_\varsigma$, holds almost surely due to lack of censoring before $t_\varsigma$ and thus can be dropped.
    In contrast, for $\varsigma > 1$, the estimand $\alpha_{*,\varsigma}$ requires taking an expectation over the unobserved distribution of history among survivors: $\alpha_{*,\varsigma} = \mathbb{E}_*[\beta_*(\tilde{H}_\varsigma) \mid T > t_\varsigma]$. A simple sample average over observed at-risk individuals ($X > t_\varsigma$), as in Step~\ref{alg step: average} of Algorithm~\ref{alg: SDR marg}, is generally biased for $\alpha_{*,\varsigma}$. The reason is that the conditional distribution $\tilde{H}_\varsigma \mid X > t_\varsigma$ differs from $\tilde{H}_\varsigma \mid T > t_\varsigma$ because of informative censoring due to covariates prior to $t_\varsigma$. Identifying $\alpha_{*,\varsigma}$ for $\varsigma > 1$ may require additional weighting to account for attrition prior to $t_\varsigma$, which is beyond the scope of this paper.
\end{remark}

The sequentially doubly robust result in Theorem~\ref{thm: bound} still applies. Moreover, under the stronger conditions below, the SDR estimator is asymptotically normal and efficient. In contrast, asymptotic normality does not hold for the G-computation or IPCW estimator in general without stringent assumptions on nuisance estimators $(\hat{S}_{k},\hat{G}_{k},\hat{U}_{k}^\SDR)$, namely them being estimated at a parametric rate.

Let $S_k$ and $G_k$ be conditional survival functions---that is, both are non-increasing, non-negative, and right-continuous, and both take value 1 at time $t_k$ given any historical covariates---such that $G_k(t_{k+1} \mid H_k)>0$.
For all $k$ and $t \in (t_k,t_{k+1}]$, define pointwise
\begin{equation}
    \bar{R}_{k,t}(S_k,G_k \mid H_k) := -S_k(t \mid H_k) \int_{(t_k,t]} \frac{G_k(s- \mid H_k)-G_{*,k}(s- \mid H_k)}{G_k(s- \mid H_k)} \left( \frac{S_k-S_{*,k}}{S_k} \right)(\intd s \mid H_k), \label{eq: DR remainder}
\end{equation}
whenever $G_k(t- \mid H_k) > 0$ and zero otherwise.

\begin{condition}[Negligible remainder] \label{condition: negligible remainder}
    For all $k$, $\| \bar{R}_{k,t_{k+1}}(\hat{S}_k,\hat{G}_k \mid \cdot) \|_k = \smallo_p(n^{-1/2})$ and $\| \hat{G}_{k} - G_{*,k} \|_{(t_k,t_{k+1}],k} \| \hat{U}_k^\SDR - \bar{U}_k^\SDR \|_{k+1} = \smallo_p(n^{-1/2})$.
\end{condition}

As we mentioned in Remark~\ref{remark: product remainder}, 
although the remainder $\bar{R}_{k,t}$ is not a product, it might vanish at a rate faster than $n^{-1/2}$ even if both $\hat{S}_{k}$ and $\hat{G}_{k}$ converge slower than $n^{-1/2}$.
The requirement that the product $\| \hat{G}_{k} - G_{*,k} \|_{(t_k,t_{k+1}],k} \| \hat{U}_k^\SDR - \bar{U}_k^\SDR \|_{k+1}$ vanishes at a rate faster than $n^{-1/2}$ is common in non- and semi-parametric inference literature \protect\citep{Ying2023}.
Therefore, Condition~\ref{condition: negligible remainder} may be plausible.

We next present a few more technical conditions. For any scalar $\alpha$, we define function $D(\{U_j, S_j, G_j\}_{j=1}^K, \allowbreak \alpha): o \mapsto \T_{1}( \{U_j, S_j, G_j\}_{j=1}^K)(o) - \alpha$.
As we will show, $D(\{U_{*,j}, S_{*,j}, G_{*,j}\}_{j=1}^K, \alpha_*)$ is the efficient influence function of $\hat{Q}_{0}^\SDR$ under a nonparametric model under conditions.
Recall that we use $\| \cdot \|_1$ to denote the marginal (at the first visit time $t_1=0$) $L^2(P_*)$-norm.

\begin{condition}[Bounded $\hat{U}_k^\SDR$] \label{condition: bounded U}
    For all $k=1,\ldots,K$, $\hat{U}_k^\SDR$ is $P_*$-almost surely bounded by some constant with probability tending to one.
\end{condition}

\begin{condition}[Consistency of influence function] \label{condition: IF consistency}
    $\| \T_1(\{\hat{U}_{j}^\SDR, \hat{S}_{j}, \hat{G}_{j}\}_{j=1}^K) - \T_1(\{U_{*,j}, S_{*,j}, G_{*,j}\}_{j=1}^K) \|_1 \allowbreak = \smallo_p(1)$.
\end{condition}

\begin{condition}[Donsker condition] \label{condition: Donsker}
    The random function $\T_1(\{\hat{U}_{j}^\SDR, \hat{S}_{j}, \hat{G}_{j}\}_{j=1}^K)$ falls in a $P_*$-Donsker class with probability tending to one.
\end{condition}

Condition~\ref{condition: bounded U} is reasonable because $\hat{U}_k^\SDR$ is an estimator of a function bounded by one.
Condition~\ref{condition: IF consistency} would often hold if nuisance functions are estimated consistently.
By the definition of $U_{*,k}$ in Theorem~\ref{thm: G identify}, Theorem~\ref{thm: SDR identify} implies that
$$\expect_*[\T_{k}( \{U_j, S_j, G_j\}_{j=k}^K)(O) \mid X>t_k, H_{k-1}] = U_{*,k-1}(H_{k-1})$$
for all $k=2,\ldots,K$ under the pattern of correct model specification in Theorem~\ref{thm: SDR identify}.
Thus, with a similar argument for Corollaty~\ref{coro: SDR}, Condition~\ref{condition: IF consistency} holds if $(\hat{S}_k,\hat{G}_k)$ is consistent and $\| \hat{U}_k^\SDR - \bar{U}_k^\SDR \|_{k+1}=\smallo_p(1)$ for all $k$. A formal argument is given in the proof of Corollary~\ref{coro: SDR} in Supplement~\ref{sec: proof bound}.
Condition~\ref{condition: Donsker} is common in non- and semi-parametric inference literature, and, as usual \protect\citep{Schick1986,newey2018cross}, can be removed by adopting full cross-fitting (see Supplement~\ref{sec: CV}).

The above conditions lead to the following result on $\hat{Q}_{0}^\SDR$ beyond the sequential double robustness shown in Theorem~\ref{thm: bound}.
We say that an estimator $\hat{\theta}$ of $\theta_*$ is asymptotically linear with influence function $f$ if $\hat{\theta} = \theta_* + n^{-1} \sum_{i=1}^n f(O_i) + \smallo_p(n^{-1/2})$, where $\expect_*[f(O)]=0$ and $\expect_*[f(O)^2] < \infty$.

\begin{theorem}[Asymptotic efficiency of $\hat{Q}_0^\SDR$] \label{thm: efficiency}
    Under Conditions~\ref{condition: cond ind cens} and \ref{condition: consistency of strong positivity}, $D(\{U_{*,j}, S_{*,j}, G_{*,j}\}_{j=1}^K, \allowbreak \alpha_*)$ is the efficient influence function for estimating $\alpha_*$ under a nonparametric model.
    Additionally, under Conditions~\ref{condition: negligible remainder}--\ref{condition: Donsker}, $\hat{Q}_0^\SDR$ is an asymptotically efficient estimator of $\alpha_*$ under a nonparametric model, that is, it is asymptotically linear with influence function $D(\{U_{*,j}, S_{*,j}, G_{*,j}\}_{j=1}^K, \allowbreak \alpha_*)$.
    Therefore, $n^{1/2} (\hat{Q}_0^\SDR - \alpha_*) \overset{d}{\to} \mathrm{N}(0,\sigma_*^2)$ where $\sigma_*^2 := \expect_* [D(\{U_{*,j}, S_{*,j}, G_{*,j}\}_{j=1}^K, \alpha_*)(O)^2]$.
\end{theorem}

An asymptotically valid Wald confidence interval for $\alpha_*$ can be constructed based on a consistent estimator of $\sigma_*^2$, for example, $n^{-1} \sum_{i=1}^n D(\{\hat{U}_j^\SDR, \hat{S}_j, \hat{G}_j\}_{j=1}^K, \hat{Q}_0^\SDR)(O_i)^2$.
An asymptotically valid Wald confidence interval for smooth transformations of $\alpha_*$, such as the average additive/multiplicative effect and vaccine efficacy in randomized controlled trials, can be constructed similarly with standard errors derived from the Delta method.
It is also possible to construct an asymptotically valid confidence band over a range of $\tau$ by showing that $\hat{Q}_0^\SDR$, which depends on $\tau$, is uniformly asymptotically linear over a range of $\tau$, under stronger uniform versions of Conditions~\ref{condition: negligible remainder}--\ref{condition: Donsker}.

\section{Simulation} \label{sec: sim}

\subsection{Discrete time} \label{sec: disc sim}

In this simulation study, we compare the performance of our proposed SDR estimator with some nonparametric competitors when time is discrete. We consider the following data-generating mechanism with $K=2$ and visit times $t_1=0$, $t_2=3$.
We generate covariate $\tilde{L}_1$ consisting of $\tilde{L}_{11} \sim \mathrm{N}(0,1)$, $\tilde{L}_{12} \sim \mathrm{Bernoulli}(0.5)$ and $\tilde{L}_{13} \sim \mathrm{N}(0,1)$ at visit time $t_1$ independently; we then generate the portion of times to event ($T_1$) or censoring ($C_1$) in the first time window $(t_1,t_2]$ such that $T_1$ and $C_1$ are independent given $\tilde{H}_1$ and are supported in $\{1,2,3,3+\}$.
Here, ``3+'' stands for no event/censoring by $t_2$.
The conditional probability mass of $T_1$ is
$\tilde{P}_*(T_1 = 3+ \mid \tilde{H}_1) = \expit(-0.2 + 1.2 \tilde{L}_{12} + 0.2 |\tilde{L}_{11}| + 0.3 \tilde{L}_{13}^2)$, $\tilde{P}_*(T_1 = 3 \mid \tilde{H}_1) = \{1-\tilde{P}_*(T_1 = 3+ \mid \tilde{H}_1)\}/2$, $\tilde{P}_*(T_1 = 1 \mid \tilde{H}_1) = \tilde{P}_*(T_1 = 2 \mid \tilde{H}_1) = \{1-\tilde{P}_*(T_1 = 3+ \mid \tilde{H}_1)\}/4$; the conditional probability mass of $C_1$ is
$\tilde{P}_*(C_1 = 3+ \mid \tilde{H}_1) = \expit(-0.2 + 1.2 \tilde{L}_{12} + 0.2 |\tilde{L}_{11}| \tilde{L}_{12})$, $\tilde{P}_*(C_1 = 3 \mid \tilde{H}_1) = \{1-\tilde{P}_*(C_1 = 3+ \mid \tilde{H}_1)\}/2$, $\tilde{P}_*(C_1 = 1 \mid \tilde{H}_1) = \tilde{P}_*(C_1 = 2 \mid \tilde{H}_1) = \{1-\tilde{P}_*(C_1 = 3+ \mid \tilde{H}_1)\}/4$.
We generate covariates $\tilde{L}_2$ consisting of $\tilde{L}_{21} \mid \tilde{H}_1 \sim \mathrm{N}(0.8 \tilde{L}_{11},0.6^2)$ and $\tilde{L}_{22} \mid \tilde{H}_1 \sim \mathrm{Bernoulli}(0.4 \tilde{L}_{12} + 0.3)$ at visit time $t_2$ independently; we then generate time to event $T_2$ after $t_2$ with conditional probability mass $\tilde{P}_*(T_2 = 0 \mid \tilde{H}_2,C_1) = \ind(T_1 \neq 3+)$, $\tilde{P}_*(T_2 = 3 \mid \tilde{H}_2,T_1) = \ind(T_1 = 3+) \expit(-0.1 + 1.2 \tilde{L}_{22} + 0.2 \tilde{L}_{21}^2 + 0.3 \tilde{L}_{13}^2)$, $\tilde{P}_*(T_2 = 2 \mid \tilde{H}_2,T_1) = \tilde{P}_*(T_2 = 1 \mid \tilde{H}_2,T_1) = \ind(T_1 = 3+) \{1-\tilde{P}_*(T_2 = 3 \mid \tilde{H}_2,T_1)\}/2$, and time to censoring $C_2$ after $t_2$ with conditional probability mass $\tilde{P}_*(C_2 = 0 \mid \tilde{H}_2,C_1) = \ind(C_1 \neq 3+)$, $\tilde{P}_*(C_2 = 3 \mid \tilde{H}_2,C_1) = \ind(C_1 = 3+) \expit(-0.1 + 1.2 \tilde{L}_{22} + 0.2 \tilde{L}_{11}^2 + 0.1 \tilde{L}_{12})$, $\tilde{P}_*(C_2 = 2 \mid \tilde{H}_2,C_1) = \tilde{P}_*(C_2 = 1 \mid \tilde{H}_2,C_1) = \ind(C_1 = 3+) \{1-\tilde{P}_*(C_2 = 3 \mid \tilde{H}_2,C_1)\}/2$.
The total times to event and censoring are $T=T_1+T_2$ and $C=C_1+C_2$, respectively, where ``3+'' is treated as 3 in the addition. The time of follow-up is $X=T \wedge C$ and the event indicator is $\Delta=\ind(T \leq C)$.
We consider $\tau=5$ and $\varsigma=1$. As mentioned below Theorem~\ref{thm: G identify}, setting $\varsigma=1$ loses no generality for the conditional estimand $\tilde{P}_*(T > \tau \mid T>t_\varsigma, \tilde{H}_\varsigma=h_\varsigma)$.

This data-generating mechanism was chosen so that the probability of censoring is moderate (39.6\%), the conditional probabilities of not being censored $G_{*,1}(t_2 \mid \tilde{H}_1)$ and $G_{*,2}(\tau \mid \tilde{H}_2)$ are effectively positive (minimum in a held-out simulated population of size $10^6$ being 0.45 and 0.48, respectively), and the marginal survival probability $\tilde{P}_*(T > \tau) = 0.49$. The distributions of $T$ and $C$ were chosen such that (i) the nuisance estimators of $(S_{*,k},G_{*,k})$ all have misspecification bias, and (ii) the historical covariates $\tilde{H}_k$ are informative of both times to event and censoring within the next time window $(t_k,t_{k+1}]$, $T_k$ and $C_k$ ($k=1,2$).
We consider both the marginal and conditional survival probabilities at $\varsigma=1$.
The sample size $n$ ranges over $\{500,1000,2000\}$.

When estimating the marginal survival probability, we compare the performance of the following nonparametric methods:
(i) sequentially doubly robust estimator, with nuisance functions estimated consistently (\texttt{SDR}), $(S_{*,k},\bar{U}_k^\SDR)$ estimated inconsistently (\texttt{SDR.SUmis}), $G_{*,k}$ estimated inconsistently (\texttt{SDR.Gmis}), $(G_{*,1},S_{*,2})$ estimated inconsistently (\texttt{SDR.mix1}), or $(S_{*,1},\bar{U}_1^\SDR,G_{*,2})$ estimated inconsistently (\texttt{SDR.mix2}), 
(ii) G-computation estimator, with $(S_{*,k},\bar{U}_k^\SDR)$ estimated consistently (\texttt{Gcomp}) or inconsistently (\texttt{Gcomp.SUmis}), or with the mixed inconsistency as for the SDR estimator (\texttt{Gcomp.mix1} and \texttt{Gcomp.mix2}),
(iii) IPCW estimator, with $G_{*,k}$ estimated consistently (\texttt{IPCW}) or inconsistently (\texttt{IPCW.Gmis}), or with the mixed inconsistency (\texttt{IPCW.mix1} and \texttt{IPCW.mix2}),
(iv) the nonparametric inference method proposed by \protect\citet{Westling2023} as implemented in the R package \texttt{CFsurvival}, and (v) the targeted minimum loss-based estimator proposed by \protect\citet{Benkeser2018} as implemented in the R package \texttt{survtmle}.

Both \texttt{CFsurvival} and \texttt{survtmle} were developed to estimate counterfactual survival probabilities under a binary treatment setting. To adapt these packages for our marginal survival setting, we assigned all individuals a dummy treatment $W=0$. While both methods adjust for baseline covariates, neither accommodates time-varying covariates. \texttt{CFsurvival} is almost identical to our proposed marginal SDR estimator in Section~\ref{sec: marg SDR estimator} for $K=1$, differing primarily in its inclusion of an inverse propensity score weight to address causal confounding. We evaluated \texttt{CFsurvival} by estimating the counterfactual survival curve corresponding to $W = 0$. \texttt{survtmle} has a further restriction that time must be discrete and models conditional hazards via supervised learning algorithms rather than algorithms for general right-censored data. Because the \texttt{survtmle} implementation requires at least one observation in each treatment arm so that inverse propensity score can be computed for both treatment arms, we appended a single synthetic placeholder observation assigned to $W = 1$ (with baseline covariates set to zero and an event at $t = 1$) to allow code execution. We then extracted the marginal counterfactual survival curve for $W = 0$. Given our large simulation sample sizes, the impact of this single placeholder observation on \texttt{survtmle} estimates is negligible.
\texttt{SDR.mix1} and \texttt{SDR.mix2} correspond to the inconsistent estimation of outcome models in one time window and the censoring model in another time window, a scenario that can only be handled by SDR estimators rather than conventional doubly robust estimators or non-robust estimators.

In all the above methods except \texttt{survtmle}, we estimate nuisance survival functions $S_{*,k}$ and $G_{*,k}$ by survival Super Learner \protect\citep{Westling2023} with a library consisting of the Kaplan-Meier estimator \protect\citep{Kaplan1958}, the Cox proportional hazard model \protect\citep{Andersen1982,Cox1972} in conjunction with the Breslow estimator \protect\citep{Breslow1972}, and random survival forests \protect\citep{Ishwaran2008} with various tuning parameters (\texttt{nodesize} $\in \{10,15,20\}$, \texttt{ntime} $\in \{150,200,250\}$ in function \texttt{survSL.rfsrc}) when we wish to estimate these functions consistently. When we estimate a nuisance survival function inconsistently, we use the Kaplan-Meier estimator, which is inconsistent because time-varying covariates are informative of censoring.
When running ordinary regressions to estimate functions such as $\bar{U}_k^\SDR$ and the conditional survival function $\beta_*$ (Step~\ref{alg step: regress} in Algorithm~\ref{alg: SDR}), we use Super Learner \protect\citep{VanderLaan2007} with a library consisting of gradient boosting \protect\citep{Mason1999,Friedman2001,Friedman2002,Chen2016} with various tuning parameters (\texttt{ntree} $\in \{100,200,400,800,1000\}$, \texttt{minobspernode} $\in \{1,5,10\}$ in function \texttt{SL.xgboost}), as well as linear models with and without interactions.
When we estimate $\bar{U}_k^\SDR$ inconsistently, we use a linear model with main effects.
When estimating the nuisance discrete-time hazards in \texttt{survtmle}, we use Super Learner with a library consisting of gradient boosting with the above tuning parameters.

We adopt 5-fold cross-fitting for SDR, G-computation and IPCW estimators, as described in Supplement~\ref{sec: CV}.
We use the default setting for \texttt{CFsurvival} and \texttt{survtmle}, except that the treatment propensity score estimator is an intercept-only regression because no real treatment is present.
We compare the sampling distributions of all these estimators for the marginal survival probability $\alpha_*$, and compare the coverage of the 95\%-Wald confidence interval (CI) for the sequentially doubly robust estimators, \texttt{CFsurvival} and \texttt{survtmle}.
For the conditional survival probability function $\beta_*$, we compare the $L^2(P_*)$-distance (based on an independent sample of size $10^6$) for all these methods except \texttt{CFsurvival} and \texttt{survtmle} because these two methods cannot estimate a conditional survival probability function.

Figures~\ref{fig: disc margp dist} and \ref{fig: disc margp CI} contain results for the marginal survival probability.
As shown in Fig.~\ref{fig: disc margp dist}, both \texttt{survtmle} and \texttt{CFsurvival} have competitive bias and variance, especially in large samples ($n=2000$), but this might be because the identification bias from not accounting for time-varying covariates is small in this simulation setting.
G-computation, IPCW and SDR estimators with consistent nuisance function estimators all appear consistent, but the SDR estimator appears to be the least biased.
When the set of nuisance estimators of time-to-event outcomes or censoring is inconsistent in both time windows, or in the two mixed inconsistency cases, G-computation and IPCW estimators also appear biased and inconsistent in at least one scenario, but the SDR estimator appears consistent across all these scenarios.
In terms of Wald CI coverage, the coverage based on the SDR estimator with consistent nuisance estimators appears to be close to the nominal coverage across all sample sizes.
\texttt{SDR.Gmis} and \texttt{SDR.mix1} appear to also maintain approximately correct coverage.
\texttt{CFsurvival}, \texttt{survtmle}, \texttt{SDR.mix2}, and \texttt{SDR.SUmis} appear moderately anti-conservative for some sample sizes.
The relatively poor coverage of \texttt{CFsurvival} and \texttt{survtmle} may be explained by not accounting for informative censoring due to time-varying covariates.
Therefore, when estimating the marginal survival probability $\alpha_*$, the SDR estimator appears to be sequentially doubly robust against inconsistent nuisance estimators as shown in Corollary~\ref{coro: SDR}, and asymptotically normal with valid Wald CI coverage when nuisance function estimators are consistent as shown in Theorem~\ref{thm: efficiency}.

\begin{figure}[h!tb]
    \centering
    \includegraphics[width=\linewidth]{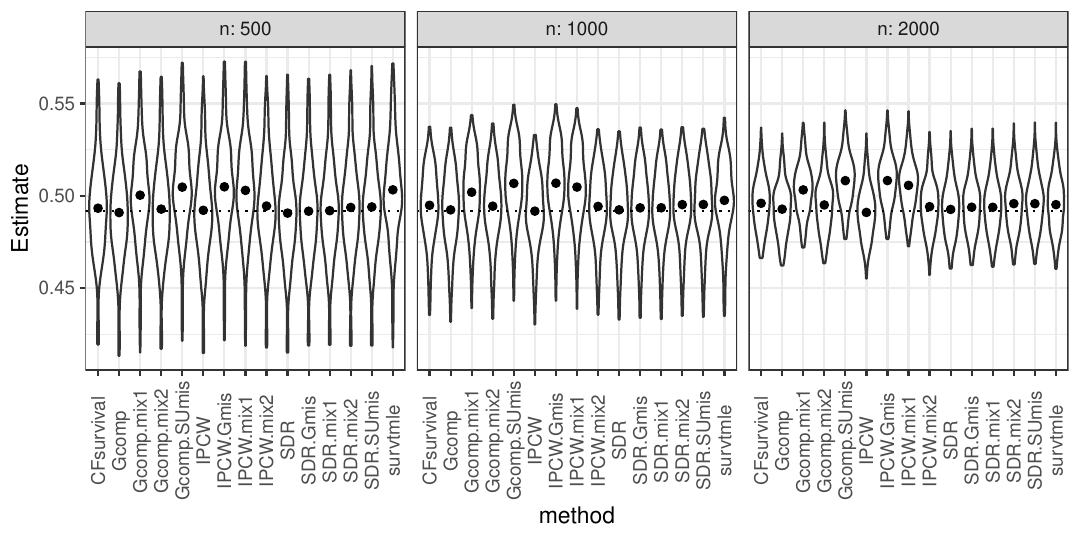}
    \caption{Sampling distribution of estimators of marginal survival probability in the discrete-time setting. The points in the middle of each violin are the estimated means of the estimators. The horizontal dotted line is the truth.}
    \label{fig: disc margp dist}
\end{figure}

\begin{figure}[h!tb]
    \centering
    \includegraphics[width=\linewidth]{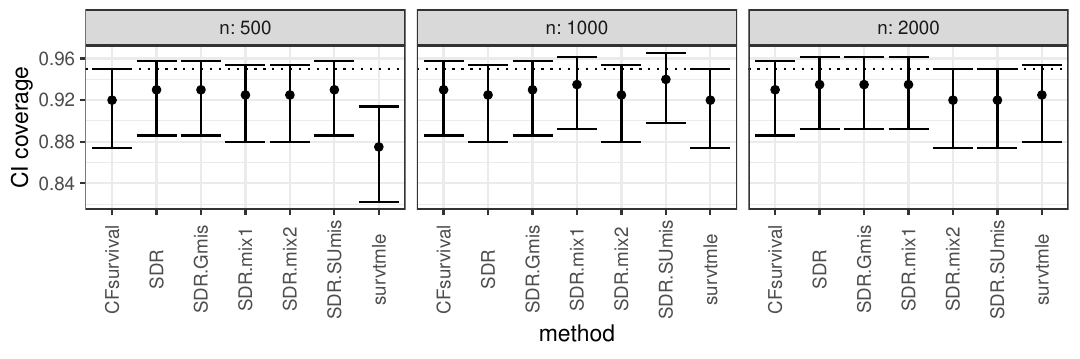}
    \caption{Estimated coverage of 95\%-Wald confidence interval (CI) for marginal survival probability, along with error bars representing 95\% confidence interval for the coverage, in the discrete-time setting. The horizontal dotted line is the nominal coverage 95\%.}
    \label{fig: disc margp CI}
\end{figure}

The $L^2(P_*)$-distance of the estimated conditional survival probability functions is presented in Fig.~\ref{fig: disc condp L2}.
For readability, we exclude the G-computation and IPCW estimators with at least one inconsistent relevant nuisance estimator in a similar figure, Fig.~\ref{fig: disc condp L22}, in the Supplemental Material.
When all nuisance estimators are consistent, the G-computation estimator appears to be closest to the truth on average, but IPCW and SDR estimators do not appear substantially worse, especially in large samples ($n=2000$).
However, when at least one relevant nuisance estimator is inconsistent, the G-computation and IPCW estimators can have a substantially larger $L^2(P_*)$-distance from the truth; in contrast, the SDR estimator still appears to be consistent.
Therefore, when estimating the conditional survival probability function, the SDR estimator appears to be sequentially doubly robust against inconsistent nuisance estimators, and to still perform reasonably well when nuisance estimators are all consistent.

\begin{figure}[h!tb]
    \centering
    \includegraphics[width=\linewidth]{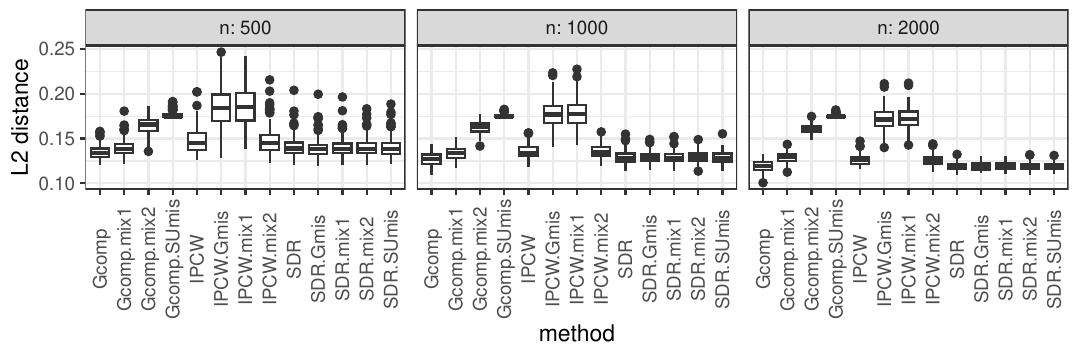}
    \caption{Sampling distribution of the $L^2(P_*)$-distance between the estimated conditional survival probability function and the truth in the discrete-time setting.}
    \label{fig: disc condp L2}
\end{figure}

\subsection{Continuous time} \label{sec: cont sim}

We conduct a simulation with continuous time similarly to the discrete-time simulation. We first describe the simulation setting, focusing on the differences.
We set $t_2=30$ and $\tau=60$. We generate the covariates in the same way and generate the times according to $T_1 \mid \tilde{H}_1 \sim \mathrm{Weibull}(5,30+20*\tilde{L}_{12}+2 | \tilde{L}_{11} |+\tilde{L}_{13}^2) \wedge (t_2-t_1)$, $C_1 \mid \tilde{H}_1 \sim \mathrm{Weibull}(4,35+15 \tilde{L}_{12}+0.5 | \tilde{L}_{11} | \tilde{L}_{12}) \wedge (t_2-t_1)$, $T_2 \mid \tilde{H}_2,T_1 \sim \ind(T_1=t_2-t_1) \mathrm{Weibull}(3,30+20*\tilde{L}_{22}+2 | \tilde{L}_{21} |+\tilde{L}_{13}^2)$, and $C_2 \mid \tilde{H}_2,C_1 \sim \ind(C_1=t_2-t_1) \mathrm{Weibull}(4,35+15 \tilde{L}_{22}+0.5 | \tilde{L}_{21} | \tilde{L}_{22}) \wedge (60-t_2)$.
Here, $\mathrm{Weibull}(a,b)$ stands for the Weibull distribution with shape parameter $a$ and scale parameter $b$, so that the mean is $b \Gamma(1+1/a)$ and the variance is $b^2 \{ \Gamma(1+2/a) - [\Gamma(1+1/a)]^2 \}$.
This data-generating mechanism was chosen so that the probability of censoring is moderate (49.3\%), the conditional probabilities of not being censored $G_{*,1}(t_2 \mid \tilde{H}_1)$ and $G_{*,2}(\tau \mid \tilde{H}_2)$ are effectively positive (minimum in a held-out simulated population of size $10^6$ being 0.58 and 0.53, respectively), and the marginal survival probability $\tilde{P}_*(T > \tau) = 0.49$.

There are also a few differences in the estimators. When using \texttt{survtmle}, because it only handles discrete times, we round all times to the nearest integer. When estimating nuisance survival functions $(S_{*,k},G_{*,k})$, we additionally include Weibull regression \protect\citep{Wei1992} in the survival Super Learner library.

The simulation results are in Fig.~\ref{fig: cor margp dist}--\ref{fig: cor condp L2}. The findings are similar to the discrete-time setting, with a notable exception that \texttt{CFsurvival} and \texttt{survtmle} have larger biases and much lower CI coverage than SDR estimators when estimating the marginal survival probability, confirming the necessity to account for time-varying covariates informative of future censoring.
The SDR estimator still performs competitively: it appears consistent across the four inconsistency patterns, and its confidence interval coverage is closest to the nominal coverage when all nuisance function estimators are consistent.

\begin{figure}[h!tb]
    \centering
    \includegraphics[width=\linewidth]{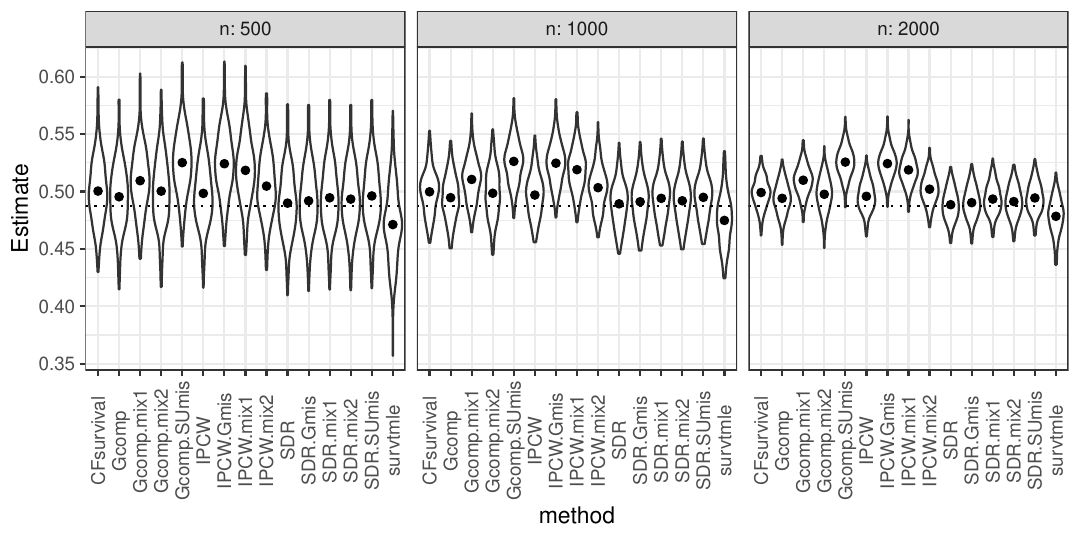}
    \caption{Similar to Fig.~\ref{fig: disc margp dist} in the continuous-time setting.}
    \label{fig: cor margp dist}
\end{figure}

\begin{figure}[h!tb]
    \centering
    \includegraphics[width=\linewidth]{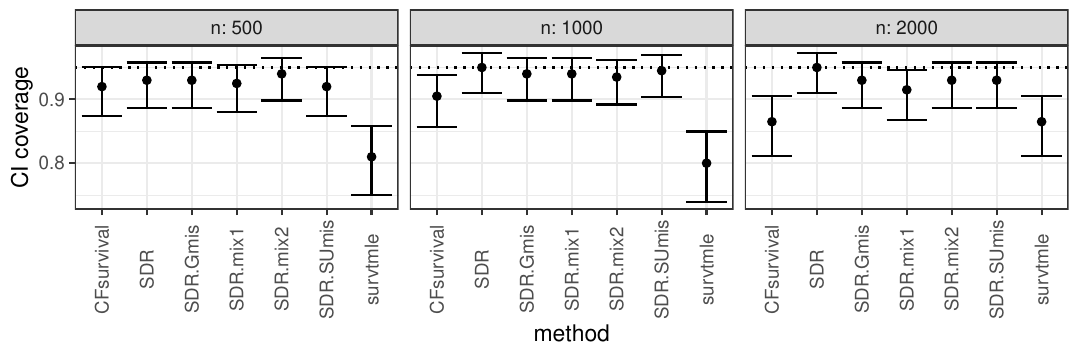}
    \caption{Similar to Fig.~\ref{fig: disc margp CI} in the continuous-time setting.}
    \label{fig: cor margp CI}
\end{figure}

\begin{figure}[h!tb]
    \centering
    \includegraphics[width=\linewidth]{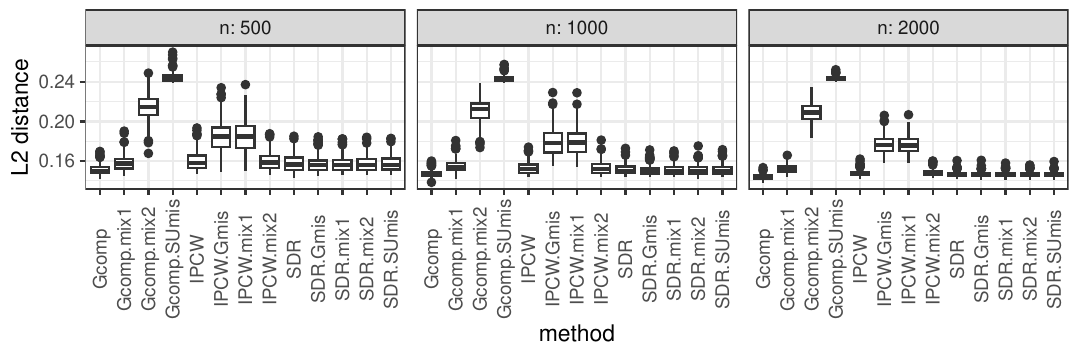}
    \caption{Similar to Fig.~\ref{fig: disc condp L2} in the continuous-time setting.}
    \label{fig: cor condp L2}
\end{figure}

\section{Application to a COVID-19 vaccine efficacy trial} \label{sec: vac trial}

We apply the proposed methods to analyze VAT00008 Stage 2 (NCT04904549) to estimate the additive effect and CDE described in Section~\ref{sec: example}, in the baseline seropositive cohort because it is most relevant for representing contemporary epidemiology.
To spare costs, a case-cohort sampling design was used to select a subset of participants for measuring the biomarkers at both time points, and we applied inverse probability of sampling weighting in the analysis.
In this cohort, a total of 146 COVID-19 endpoints were observed by 365 days post Day~22 visit.
Baseline covariates $W_c$ include sex, force of infection score that estimates from external data bases the amount of SARS-CoV-2 transmission happening in a participant’s local geography and time of follow-up, and standardized baseline risk score.

We adopt 20-fold cross-fitting.
We estimate nuisance survival functions by survival Super Learner with a library consisting of the Kaplan-Meier estimator, the Cox proportional hazard model, piecewise constant hazard models \protect\citep[e.g.,][]{Andersen2021} with the number of breaks ranging over $\{4,6,8\}$, and random survival forests with tuning parameters in Section~\ref{sec: disc sim}.
When estimating $U_{*,k}$ via regression, we use Super Learner with a library consisting of gradient boosting and neural networks with various tuning parameters, as well as linear models with and without interactions.
When estimating $\beta_*$ by regressing pseudo-outcomes, we only include linear models in Super Learner to avoid overfitting.
After running ordinary regressions with pseudo-outcomes, we additionally project estimates of $\beta_{*}(H_\varsigma)$ onto the unit interval and run an isotonic regression over $\tau$.
We estimate CDE by plugging in the estimates of the conditional survival probability.

The estimated additive effects of the vaccine versus placebo on the marginal cumulative incidence are presented in Table~\ref{tab: marg nAb}. The vaccine appears protective, consistent with the original analysis.
The vaccine effect is significant at $\tau=27, 90, 135$ but not at $\tau=180,365$.
This loss of significance likely reflects increased variance at later times due to a diminishing number of informative observations. The point estimates indicate that the additive-difference vaccine effect follows a non-linear trend---initially increasing as the immune response matures before undergoing a gradual attrition over the long term.
Fig.~\ref{fig: surv nAb} presents the estimated conditional survival probability over a range of Day~22 neutralizing antibody levels.
The overall cumulative incidence is low ($<2\%$). Males (Sex=1) have lower survival curves, namely higher incidence.
According to the SDR estimator, the Day~22 neutralizing antibody level is slightly positively correlated with survival at all times.
The estimates from G-computation and SDR estimation are somewhat similar and different from IPCW estimation. The reason might be that the random variable in the mean in Corollary~\ref{coro: IPCW identify} is non-zero only when that individual is still at risk at $t_2$ ($X>t_2$), so the IPCW estimator relies more on incidence after $t_2$ but does not fully leverage information in the survival function $S_{*,1}$ in $(t_1,t_2]$, and may thus be less stable than G-computation and SDR estimators when $\hat{S}_1$ is close to $S_{*,1}$.

Fig.~\ref{fig: CDE nAb} presents the estimated CDE curve for the neutralization antibody titer biomarker, the difference between vaccine and placebo curves in Fig.~\ref{fig: surv nAb}.
All CDE curves are close to zero, suggesting that the vaccine's effect can be well explained by its effect on neutralizing antibodies.
CDEs are negative for most values of Day~22 neutralizing antibody levels except for the SDR estimators for males (Sex=1).
This suggests the possible existence of other immune functions not captured by neutralizing antibodies or a lack of sensitivity of the neutralizing antibody assay.
In the exceptional case, CDEs are positive and further from zero for larger $\tau$.
This is reasonable because the quality of the surrogate (Day~22 neutralizing antibody level) is expected to weaken over time, and there may be other immune functions not captured by neutralizing antibodies with a longer-term effect.
However, this study may be underpowered to detect non-zero CDE, and the conclusion above should be considered exploratory.

\begin{table}[h!tb]
    \centering
    \caption{SDR estimate (S.E.) of the vaccine's additive effect versus placebo on the cumulative incidence of the COVID-19 endpoint by time $\tau$, multiplied by 100, accounting for potential informative censoring due to covariates and neutralizing antibody biomarker measured on Day~22 and 43 visits.}
    \label{tab: marg nAb}
    \begin{tabular}{c|r|r|r|r|r}
    \hline\hline
         Time $\tau$ & 27 & 90 & 135 & 180 & 365 \\
         \hline
         Est (S.E.) & -0.64 (0.21) & -0.96 (0.31) & -1.00 (0.37) & -0.85 (0.43) & -0.57 (0.52) \\
    \hline\hline
    \end{tabular}
\end{table}

\begin{figure}[h!tb]
    \centering
    \includegraphics[width=\linewidth]{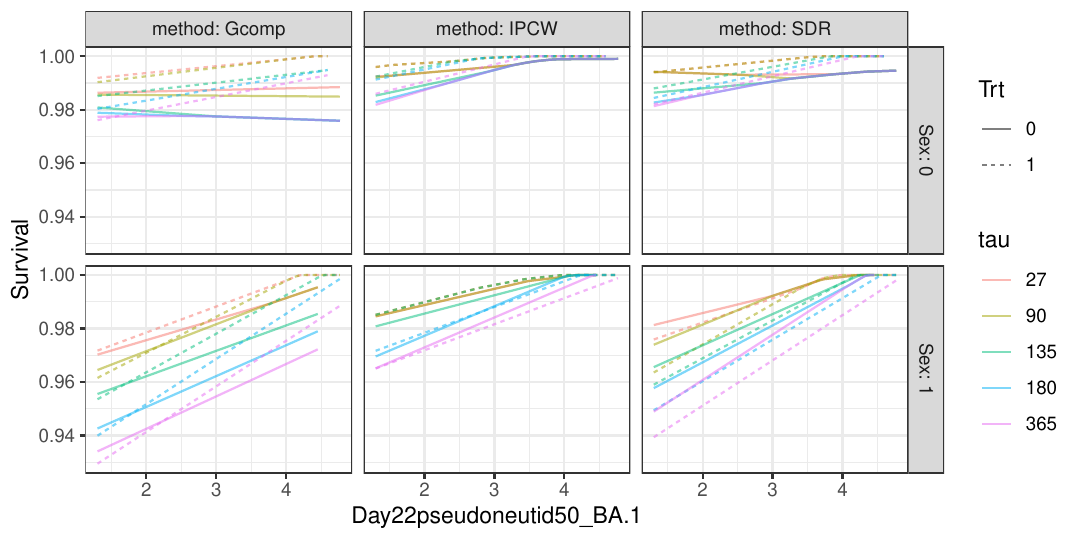}
    \caption{Estimated conditional survival probability $\beta_{*}(w)$ over the observed range of Day~22 neutralizing antibody level (\texttt{Day22pseudoneutid50\textunderscore BA.1}) and a range of time points of interest $\tau$. The range of antibody levels is stratified by sex and treatment arm (\texttt{Trt}). Force of infection score and standardized baseline risk score are fixed at the sample mean stratified by sex. Sex is coded as 0=female and 1=male.}
    \label{fig: surv nAb}
\end{figure}

\begin{figure}[h!tb]
    \centering
    \includegraphics[width=\linewidth]{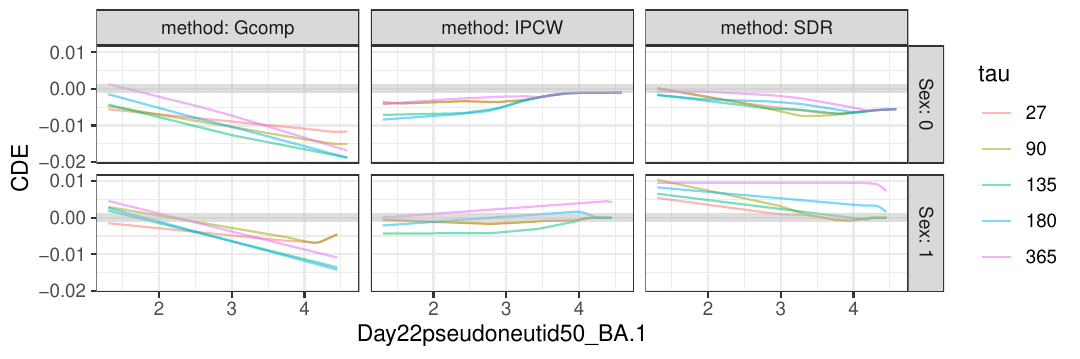}
    \caption{Estimated controlled direct effect (CDE) over the observed common range of Day~22 neutralizing antibody level (\texttt{Day22pseudoneutid50\textunderscore BA.1}) shared by the vaccine and placebo arms and a range of time points of interest $\tau$. The range of antibody levels is stratified by sex. Force of infection score and standardized baseline risk score are fixed at the sample mean stratified by sex.}
    \label{fig: CDE nAb}
\end{figure}

We also estimate the log multiplicative CDE that measures controlled direct effects in terms of multiplicative vaccine-reduction in risk, and found similar results.
The results for the IgG Spike binding antibody biomarker are similar to those for the neutralizing antibody biomarker. These additional analysis results can be found in Supplement~\ref{sec: vac trial2}.

\section{Discussion} \label{sec: discussion}

Our proposed methods have several limitations and potential future research directions.
First, the statistical problem we study in this paper does not involve non-randomized treatment assignment and the induced causal inference problem. Therefore, our methods cannot be directly applied to observational studies with confounders to study causal effects.
Adding inverse propensity score weighting \protect\citep[e.g.,][]{Hubbard2000,VanderLaan2002,Westling2023} would provide one approach.
Our theoretical results such as the novel SDR transformation might also help develop future methods for mediation analysis with right-censored outcomes.
Second, we adopt SDR transformations based on one-step correction. Targeted minimum loss-based estimation \protect\citep{VanderLaan2018} is another popular framework to achieve SDR, and might be of future interest.
Third, we do not provide statistical inference (e.g., confidence intervals) for the conditional survival function. For certain regression estimators in Step~\ref{alg step: regress} in Algorithm~\ref{alg: SDR}, such as generalized linear models and series regression, it is possible to obtain valid inference based on the SDR property of the SDR transformation and standard arguments for these estimators, under appropriate correct specification or smoothness conditions.
Fourth, we do not handle left-truncated data, but a modification of our methods similar to \protect\citet{Wang2024} might suffice.
Fifth, as discussed in Remark~\ref{rmk: marginal for varsigma>1}, our methods do not estimate the survival probability $\tilde{P}_*(T > \tau \mid T>t_\varsigma)$ unconditional on covariates at $\varsigma>1$. We leave this problem for future work.
Finally, our theoretical guarantees are asymptotic. Characterizing the finite-sample performance is an interesting future direction.
In summary, our methods may help scientists investigate treatment mechanisms on continuous time-to-event outcomes in the future.

\section*{Acknowledgement}

This project has been funded in whole or in part with federal funds from the National Institutes of Health under award numbers DP2-LM013340 and UM1AI068635 from the National Institute of Allergies and Infectious Diseases, National Institutes of Health.  
The VAT00008 data were generated from federal funds from the Department of Health and Human Services and from Sanofi.
The content is solely the responsibility of the authors and does not necessarily represent the official views of the National Institutes of Health. 
This work was supported in part through computational resources and services provided by the Institute for Cyber-Enabled Research at Michigan State University.
We thank Larry Han and Yuyao Wang for providing helpful feedback on an earlier version of this manuscript.

\clearpage

\bibliographystyle{chicago}
\bibliography{reference}

\clearpage

\setcounter{page}{1}
\setcounter{section}{0}
\renewcommand{\thepage}{S\arabic{page}}%
\renewcommand{\thesection}{S\arabic{section}}%
\renewcommand{\theHsection}{S\thesection}%
\setcounter{table}{0}
\renewcommand{\thetable}{S\arabic{table}}%
\renewcommand{\theHtable}{S\thetable}%
\setcounter{figure}{0}
\renewcommand{\thefigure}{S\arabic{figure}}%
\renewcommand{\theHfigure}{S\thefigure}%
\setcounter{equation}{0}
\renewcommand{\theequation}{S\arabic{equation}}%
\renewcommand{\theHequation}{S\theequation}%
\setcounter{condition}{0}
\renewcommand{\thecondition}{S\arabic{condition}}%
\renewcommand{\theHcondition}{S\thecondition}%
\setcounter{lemma}{0}
\renewcommand{\thelemma}{S\arabic{lemma}}%
\renewcommand{\theHlemma}{S\thelemma}%
\setcounter{theorem}{0}
\renewcommand{\thetheorem}{S\arabic{theorem}}%
\renewcommand{\theHtheorem}{S\thetheorem}%
\setcounter{corollary}{0}
\renewcommand{\thecorollary}{S\arabic{corollary}}%
\renewcommand{\theHcorollary}{S\thecorollary}%
\setcounter{algorithm}{0}
\renewcommand{\thealgorithm}{S\arabic{algorithm}}%
\renewcommand{\theHalgorithm}{S\thealgorithm}%

\begin{center}
    \LARGE Supplement to ``\ourtitle''
\end{center}

\section{G-computation and IPCW estimators} \label{sec: G IPCW alg}

The following Algorithms~\ref{alg: Gcomp} and \ref{alg: IPCW} describe G-computation and IPCW estimators motivated by Theorem~\ref{thm: G identify} and Corollary~\ref{coro: IPCW identify}, respectively.

\begin{algorithm}
    \caption{G-computation estimator of $\beta_{*}$} \label{alg: Gcomp}
    \begin{algorithmic}[1]
        \Require Data $\{O_1,\ldots,O_n\}$, algorithm to estimate $S_{*,k}$ ($k=\varsigma,\ldots,K$).
        \State Estimate $S_{*,k}$ as in Line~\ref{alg step: est S G} of Algorithm~\ref{alg: SDR}.
        \State Set $\hat{U}_{K}^\G =1$.
        \For{$k=K,\ldots,\varsigma$}
        \State Compute the pseudo-outcome
        \begin{equation}
            Y^\G_{k} := \hat{S}_{k}(t_{k+1} \mid H_k) \hat{U}_{k}^\G(H_k) \label{eq: G pseudo-outcome}
        \end{equation}
        for each individual at risk at visit time $t_k$ ($X>t_k$).
        \If{$k>\varsigma$}
            \State Regress $Y^\G_{k}$ on $H_{k-1}$ among individuals with $X>t_k$. The fitted predictive model $\hat{U}_{k-1}^\G$ is the estimator of $U_{*,k-1}$.
        \Else
            \State Regress $Y^\G_{k}$ on $H_k$ among individuals with $X>t_k$. The fitted predictive model $\hat{Q}_{k}^\G$ is the estimator of $Q_{*,k}$.
        \EndIf
        \EndFor
        \State\Return $\hat{Q}_{\varsigma}^\G$.
    \end{algorithmic}
\end{algorithm}

\begin{algorithm}
    \caption{IPCW estimator of $\beta_{*}$} \label{alg: IPCW}
    \begin{algorithmic}[1]
        \Require Data $\{O_1,\ldots,O_n\}$, algorithm to estimate $G_{*,k}$ ($k=\varsigma,\ldots,K$).
        \State Estimate $G_{*,k}$ as in Line~\ref{alg step: est S G} of Algorithm~\ref{alg: SDR}.
        \State Compute the pseudo-outcome
        \begin{equation}
            Y^\IPCW_{\varsigma} := \ind(X>t_K) \left\{ \prod_{j=\varsigma}^{K-1} \frac{1}{\hat{G}_j(t_{j+1} \mid H_j)} \right\} \left\{ 1 - \frac{\ind(X \leq \tau) \Delta}{\hat{G}_K(X- \mid H_K)} \right\} \label{eq: IPCW pseudo-outcome}
        \end{equation}
        for each individual at risk at visit time $t_\varsigma$ ($X>t_\varsigma$).
        \State Regress $Y^\IPCW_{\varsigma}$ on $H_\varsigma$ among individuals with $X>t_\varsigma$. The fitted predictive model $\hat{Q}_{\varsigma}^\IPCW$ is the estimator of $Q_{*,\varsigma}$. \label{alg step: IPCW regress}
        \State\Return $\hat{Q}_{\varsigma}^\IPCW$.
    \end{algorithmic}
\end{algorithm}

We formally define the oracle estimators for Algorithms~\ref{alg: Gcomp} and \ref{alg: IPCW} that are used in Theorem~\ref{thm: bound}:
\begin{align}
    &\bar{U}_{K}^\G := 1, \nonumber \\
    &\bar{U}_{k}^\G: h_k \mapsto \expect_*[Y^\G_{k} \mid X>t_{k+1}, H_k=h_k], \quad (k=K-1,\ldots,\varsigma) \nonumber \\
    &\bar{Q}_{\varsigma}^\G: h_\varsigma \mapsto \expect_*[Y^\G_{\varsigma} \mid X>t_\varsigma, H_\varsigma=h_\varsigma] = \hat{S}_{\varsigma}(t_{\varsigma+1} \mid h_\varsigma) \hat{U}_{\varsigma}^\G(h_\varsigma), \nonumber \\
    &\bar{Q}_{\varsigma}^\IPCW: h_\varsigma \mapsto \expect_*[Y^\IPCW_{\varsigma} \mid X>t_\varsigma, H_\varsigma=h_\varsigma]. \label{eq: G IPCW oracle}
\end{align}

\section{Cross-fitting} \label{sec: CV}

One potential issue with Algorithms~\ref{alg: SDR}, \ref{alg: Gcomp} and \ref{alg: IPCW} is that, when computing the pseudo-outcomes, the nuisance function estimators are evaluated over the same entire data as the data point $O$, which may lead to overfitting.
We illustrate with the SDR estimator.
Note that $\T_k$ can be equivalently defined recursively from $k=K$ to $k=1$ by
\begin{align}
    \T_{K}(\{U_K,S_K,G_K\})(O) &= \TS_{K,t_{K+1}}(S_K,G_K)(X,\Delta \mid H_K), \nonumber \\
    \T_{k}( \{U_j, S_j, G_j\}_{j=k}^K)(O) &= \frac{\ind(X>t_{k+1})}{G_k(t_{k+1} \mid H_k)} \{ \T_{k+1}( \{U_j, S_j, G_j\}_{j=k+1}^K)(O) - U_k(H_k) \} \nonumber \\
    &\quad + U_k(H_k) \TS_{k,t_{k+1}}(S_k,G_k)(X,\Delta \mid H_k). \label{eq: SDR transform recursive}
\end{align}
Without sample-splitting, $\hat{U}_{k}^\SDR$ is fit to the pseudo-outcome $\T_{k+1}( \{\hat{U}_j^\SDR, \hat{S}_j, \hat{G}_j\}_{j=k+1}^K)(O)$ for $O$ over the entire data, so the magnitude of the oracle residual $\T_{k+1}( \{\hat{U}^\SDR_j, \hat{S}_j, \hat{G}_j\}_{j=k+1}^K)(O) - \bar{U}_{k}^\SDR(H_k)$ may be underestimated when evaluating $\T_{k}( \{\hat{U}_j^\SDR, \hat{S}_j, \hat{G}_j\}_{j=k}^K)(O)$.
In the extreme case where $\hat{U}_{k}^\SDR$ interpolates the data---for example, regression is done via a deep neural network---even if the estimator maintains good out-of-sample performance, all residuals are zero and do not resemble the oracle residual $\T_{k+1}( \{\hat{U}^\SDR_j, \hat{S}_j, \hat{G}_j\}_{j=k+1}^K)(O) - \bar{U}_{k}^\SDR(H_{k})$.

When the regression estimator is an empirical risk minimizer over a function class, it is often required that the function class has restricted complexity, namely $P_*$-Donsker, to ensure that the regression estimator is close to the corresponding oracle estimator.
To completely drop the Donsker condition, a full cross-fitting strategy similar to \protect\citetsupp{Bodory2022} can be adopted: Randomly split the data into $2(K-\varsigma+1)$ folds of equal size, use the first fold to estimate $(S_{*,K},G_{*,K})$, use the empirical distribution in the second fold to estimate $U_{*,K-1}$, use the third fold to estimate $(S_{*,K-1},G_{*,K-1})$, use the empirical distribution in the fourth fold to estimate $U_{*,K-2}$, \textellipsis, use the $(2K-2\varsigma+1)$th fold to estimate $(S_{*,\varsigma},G_{*,\varsigma})$, and finally use the empirical distribution in the $(2K-2\varsigma+2)$th fold to estimate $Q_{*,\varsigma}$.
Then we may swap the roles of these folds and use the average of $Q_{*,\varsigma}$ estimators over all folds as the final estimator, so that the entire data is used for estimating $Q_{*,\varsigma}$.

However, this full cross-fitting approach may still be suboptimal in practice because, for each fold, each nuisance function estimator is computed using merely $n/(2K-2\varsigma+2)$ individuals.
With $K=2$ and $\varsigma=1$, this number is $n/4$, and can drop dramatically as $K$ increases.
Such a training data size might be too small with a moderate total sample size $n$ such as $\leq 1000$, because (i) the outcomes are right-censored and may be relatively rare in each time window $(t_k,t_{k+1}]$, and (ii) each regression only uses a subset of data ($X>t_k$), so the actual training data size can be much smaller than the fold size, particularly for large $k$ under heavy censoring.
This training data size issue is more concerning if cross-validation is adopted in regression (e.g., Super Learner) and may yield highly unstable nuisance function estimators.

Since a major issue with not splitting data is potential over-fitting, we adopt a partial cross-fit strategy that avoids such over-fitting in all regressions and uses more data to estimate nuisance functions.
Essentially, for each regression to estimate $U_{*,k}$ or $Q_{*,\varsigma}$, we compute the pseudo-outcome in a sample-splitting fashion, that is, we compute and evaluate nuisance estimators in different folds of the data.
With data split into $M$ folds of equal size, the training data size for estimating each nuisance function is $n (M-1)/M$ rather than $n/(2K-2\varsigma+2)$ (without accounting for the sample size reduction due to $X>t_k$).
Hence, the training data size is at least $n/2$ and may substantially increase with a larger $M$.
We describe our partially cross-fit version of Algorithm~\ref{alg: SDR} in the following Algorithm~\ref{alg: CV SDR}.
We drop the superscript $\SDR$ from the notation for conciseness. For each individual $i$, we use $H_{k,i}$ to denote the covariate $H_k$ of individual $i$.
We also adopt a similar partial cross-fitting for Algorithms~\ref{alg: Gcomp} and \ref{alg: IPCW}.

\begin{algorithm}[tb]
    \caption{Partially cross-fit sequentially doubly robust estimator of $\beta_{*}$} \label{alg: CV SDR}
    \begin{algorithmic}[1]
        \Require Data $\{O_1,\ldots,O_n\}$, algorithm to estimate survival functions $S_{*,k}$ and $G_{*,k}$ ($k=\varsigma,\ldots,K$), regression algorithm to estimate conditional mean, number $M \geq 2$ of folds.
        \State Randomly split data into $M$ folds of equal size. Let $m_i$ denote the fold containing this individual $i$ ($i=1,\ldots,n$).
        \State For each fold $m=1,\ldots,M$ and each $k=\varsigma,\ldots,K$, use data out of fold $m$ to estimate $(S_{*,k},G_{*,k})$ as in Line~\ref{alg step: est S G} of Algorithm~\ref{alg: SDR} and denote the estimator by $(\hat{S}_k^{(-m)},\hat{G}_k^{(-m)})$.
        \State Set $\hat{U}_{K}^{(-m)} =1$ for all folds $m=1,\ldots,M$.
        \For{$k=K,\ldots,\varsigma$}
        \State For each individual $i$ such that $X_i > t_k$, compute the pseudo-outcome $Y_{k,i} := \T_{k}( \{\hat{U}_{j}^{(-m_i)}, \hat{S}_j^{(-m_i)}, \hat{G}_j^{(-m_i)}\}_{j=k}^K)(O_i)$.
        \If{$k>\varsigma$}
            \State For each fold $m$, among individuals $i$ out of fold $m$ such that $X_i > t_k$, regress $Y_{k,i}$ on $H_{k-1,i}$. Denote the fitted predictive model by $\hat{U}_{k-1}^{(-m)}$.
        \Else
            \State Among all individuals $i$ such that $X_i > t_k$, regress $Y_{k,i}$ on $H_{k,i}$. Denote the fitted predictive model by $\hat{Q}_{k}$.
        \EndIf
        \EndFor
        \State\Return $\hat{Q}_{\varsigma}$.
    \end{algorithmic}
\end{algorithm}

Our partial cross-fitting keeps subtle dependence among nuisance estimators and the empirical distribution. Strictly speaking, our partial cross-fitting cannot drop Donsker conditions, even if $\hat{Q}_\varsigma$ is obtained by averaging $M$ (partial) sample-split estimators of $Q_{*,\varsigma}$ in the last step of Algorithm~\ref{alg: CV SDR}.
Consider the case where $\varsigma=1$, $K=2$ and $M=2$ as an example.
We use fold~2 to estimate $(\hat{S}_2^{(-1)},\hat{G}_2^{(-1)})$, and use pseudo-outcomes in fold~1 to estimate $\hat{U}_1^{(-2)}$.
Thus, the pseudo-outcomes $Y_{2,i}$ in fold~1 depend on the entire data, and so does $\hat{U}_1^{(-2)}$.
The same argument with the two folds' roles swapped yields that $\hat{U}_1^{(-1)}$ also depends on the entire data.
Since the pseudo-outcome for estimating $Q_{*,1}$ must depend on $\hat{U}_1^{(-1)}$ or $\hat{U}_1^{(-2)}$, the empirical distribution and the nuisance functions in this regression must be dependent.
Therefore, Donsker conditions cannot be dropped by the usual argument of conditioning on some data \protect\citepsupp[e.g.,][]{Chernozhukov2018}.
In contrast, in the aforementioned full cross-fitting, a separate fold is always held out to evaluate nuisance functions, and so the empirical distribution and the nuisance functions in each regression are independent.
Still, our estimator in Algorithm~\ref{alg: CV SDR} performs well in the simulation in Section~\ref{sec: sim} even though we use flexible machine learning algorithms such as random survival forest and gradient boosting.

\section{Product bias remainder for discrete times} \label{sec: product remainder}

The terms of the form $\| \hat{S}_{k} - S_{*,k} \|_{(t_k,t_{k+1}],k} \wedge \| \hat{G}_{k} - G_{*,k} \|_{(t_k,t_{k+1}],k}$ in the bound \eqref{eq: SDR point bound} arises from bounding $\bar{R}_{k,t_{k+1}}$ defined in \eqref{eq: DR remainder}.
In this section, we study the form of the remainder $\bar{R}_{k,t_{k+1}}$ when both times to event and censoring are discrete.
Without loss of generality, we consider $k=K=1$. We assume that the support of $T$ and $C$ is contained in $\{t^j: j=1,\ldots,J+1\}$ where $t^0:=0<t^1<t^2<\ldots<t^J < t^{J+1}:=\infty$ and the time of interest $\tau = t_2 \in [t^J, \infty)$. 
The subscript $k=1$ as well as the subscript $t_2$ of $\bar{R}$ is dropped in this section.
The assumption that $t^J \leq t_2$ is innocuous since only times before $t_2$ need to be considered.
Let $\Lambda_*$ be the true cumulative hazard function and $\hat{\Lambda}$ be the estimated cumulative hazard function associated with $\hat{S}$. We use $\lambda_*$ and $\hat{\lambda}$ to denote the hazard function associated with $\Lambda_*$ and $\hat{\Lambda}$, respectively. Then it holds that
$$S_*(t \mid h) = \prod_{t^j \leq t} (1-\lambda_*(t^j \mid h)), \quad \hat{S}(t \mid h) = \prod_{t^j \leq t} (1-\hat{\lambda}(t^j \mid h)),$$
and consequently
\begin{align*}
    \bar{R}(\hat{S},\hat{G} \mid h)
    &:= -\hat{S}(t_2 \mid h) \int_{(0,t_2]} \frac{\hat{G}(s- \mid h) - G_*(s- \mid h)}{\hat{G}(s- \mid h)} \left( \frac{\hat{S}-S_*}{\hat{S}} \right)(\intd s \mid h) \\
    &\ = -\hat{S}(t_2 \mid h) \sum_{j=1}^J \frac{\hat{G}(t^{j-1} \mid h) - G_*(t^{j-1} \mid h)}{\hat{G}(t^{j-1} \mid h)} \\
    &\ \quad\times \left\{ \frac{\hat{S}(t^{j-1} \mid h)[1-\hat{\lambda}(t^j \mid h)] - S_*(t^{j-1} \mid h)[1-\lambda_*(t^j \mid h)]}{\hat{S}(t^{j-1} \mid h)[1-\hat{\lambda}(t^j \mid h)]} - \frac{\hat{S}(t^{j-1} \mid h) - S_*(t^{j-1} \mid h)}{\hat{S}(t^{j-1} \mid h)} \right\} \\
    &\ = \hat{S}(t_2 \mid h) \sum_{j=1}^J \frac{\hat{G}(t^{j-1} \mid h) - G_*(t^{j-1} \mid h)}{\hat{G}(t^{j-1} \mid h)} \frac{S_*(t^{j-1} \mid h) [\hat{\lambda}(t^j \mid h) - \lambda_*(t^j \mid h)]}{\hat{S}(t^j \mid h)}.
\end{align*}
Under Condition~\ref{condition: consistency of strong positivity}, since $G_*,\hat{G},S_*,\hat{S},\lambda_*,\hat{\lambda} \in [0,1]$, we may apply Cauchy-Schwarz inequality to show that
\begin{align*}
    | R(\hat{S},\hat{G} \mid h) | &\lesssim \left\{ \sum_{j=1}^J \left[ \hat{G}(t^{j-1} \mid h) - G_*(t^{j-1} \mid h) \right]^2 \right\}^{1/2} \left\{ \sum_{j=1}^J \left[ \hat{\lambda}(t^j \mid h) - \lambda_*(t^j \mid h) \right]^2 \right\}^{1/2} \\
    &\lesssim \sup_{t \in (0,t_2]} \left| \hat{G}(t \mid h) - G_*(t \mid h) \right| \sup_{t \in (0,t_2]} \left| \hat{\lambda}(t \mid h) - \lambda_*(t \mid h) \right|
\end{align*}
and hence, using the boundedness of these functions, conclude that
$$\| R(\hat{S},\hat{G} \mid \cdot) \| \lesssim \| \hat{G} - G_* \|_{(0,t_2]} \| \hat{\lambda} - \lambda_* \|_{(0,t_2]},$$
where we recall that $\| \cdot \| = \| \cdot \|_1$ is the $L^2(P_*)$-norm (conditional on $X>t_1$), and $\|f\|_{(0,t_2]} = \|f\|_{(0,t_2],k=1}$ is the $L^2(P_*)$-norm of $\sup_{t \in (0,t_2]} |f(t,H)|$.

The following lemma further implies that
$$\| \bar{R}(\hat{S},\hat{G} \mid \cdot) \| \lesssim \| \hat{G} - G_* \|_{(0,t_2]} \| \hat{S} - S_* \|_{(0,t_2]}$$
under conditions.
Therefore, the terms of the form $\| \hat{S}_{k} - S_{*,k} \|_{(t_k,t_{k+1}],k} \wedge \| \hat{G}_{k} - G_{*,k} \|_{(t_k,t_{k+1}],k}$ in the bound \eqref{eq: SDR point bound} can be strengthened to a product $\| \hat{S}_{k} - S_{*,k} \|_{(t_k,t_{k+1}],k} \| \hat{G}_{k} - G_{*,k} \|_{(t_k,t_{k+1}],k}$ when the time is discrete.

\begin{lemma}[Equivalent difference for hazard and survival functions for discrete time] \label{lemma: equivalent difference for hazard and survival functions}
    Consider the setting in Section~\ref{sec: product remainder}. For each $j=1,\ldots,J$, it holds that
    $$\left| \hat{S}(t^j \mid h_k) - S_*(t^j \mid h_k) \right| \lesssim \sup_{t \in (0,t_2]} \left| \hat{\lambda}(t \mid h_k) - \lambda_*(t \mid h_k) \right|.$$
    If $S_*(t_2 \mid \cdot)$ is bounded away from zero, it also holds that
    $$\left| \hat{\lambda}(t^j \mid h_k) - \lambda_*(t^j \mid h_k) \right| \lesssim \sup_{t \in (0,t_2]} \left| \hat{S}(t \mid h_k) - S_*(t \mid h_k) \right|$$
    where the constant in $\lesssim$ depends on $S_*(t_2 \mid h_k)$.
\end{lemma}

The assumption that $S_*$ is bounded away from zero can likely be dropped or relaxed using a tighter bound accounting for the multiplicative factor $S_*$ or $\hat{S}$. We made this assumption for convenience and do not investigate a tighter bound here.

\section{Additional theoretical results and proof} \label{sec: proof}

\subsection{Identification} \label{sec: proof identify}

It will be convenient to consider time windows $(t_k,t_{k+1}]$ between adjacent visit times $t_k$ and $t_{k+1}$.
For each $k \in \{1,\ldots,K\}$, we define $T_k := \{(T \wedge t_{k+1}) - t_k\} \ind(T > t_k)$, $C_k := \{(C \wedge t_{k+1}) - t_k\} \ind(C>t_k)$, $X_k := T_k \wedge C_k$, and $\Delta_k := \ind(0 < T-t_k = T_k \leq C_k)$. In other words, $T_k$ is the amount of time before the event in the window $(t_k,t_{k+1}]$ and is set to zero if $T \leq t_k$; $C_k$, $X_k$ and $\Delta_k$ are similar counterparts of $C$, $X$ and $\Delta$ for the window $(t_k,t_{k+1}]$. It will be helpful for understanding our results to note the following facts: (i) $T_k < t_{k+1} - t_k$ implies $T_{k'} = 0$ for all $k'>k$, (ii) $C_k < t_{k+1} - t_k$ implies $C_{k'} = 0$ for all $k'>k$, and (iii) $T=\sum_{k=1}^{K} T_k$, $C=\sum_{k=1}^{K} C_k$, $X=\sum_{k=1}^{K} X_k$, $\Delta= \sum_{k=1}^{K} \Delta_k$.

\begin{proof}[Proof of Lemma~\ref{lemma: at-risk set}]
    Let $\overset{d}{=}$ denote that two random vectors are identically distributed. 
    Since $T \wedge \tau - t_k$ has a one-one correspondence with $(T_k,\ldots,T_K)$ if $T>t_k$, it suffices to prove that
    $$\left\{ (\tilde{L}_k,\ldots,\tilde{L}_K,T_k,\ldots,T_K) \mid T>t_k,\tilde{H}_{k-1}=h_{k-1} \right\} \overset{d}{=} \left\{ (\tilde{L}_k,\ldots,\tilde{L}_K,T_k,\ldots,T_K) \mid X>t_k,H_{k-1}=h_{k-1} \right\}.$$
    Let $h_{k-1} = (\ell_1,\ldots,\ell_{k-1})$.
    \begin{align*}
        &(\tilde{L}_k,\ldots,\tilde{L}_K,T_k,\ldots,T_K) \mid T>t_k,\tilde{H}_{k-1}=h_{k-1} \\
        &= (\tilde{L}_k,\ldots,\tilde{L}_K,T_k,\ldots,T_K) \mid T_{j}=t_{j+1}-t_j, \tilde{L}_j=\ell_j \text{ for all } j < k \\
        \intertext{By Condition~\ref{condition: cond ind cens} at $t_1$, the display continues as}
        &\overset{d}{=} (\tilde{L}_k,\ldots,\tilde{L}_K,T_k,\ldots,T_K) \mid T_{j}=t_{j+1}-t_j, \tilde{L}_j=\ell_j \text{ for all } j < k, C_1=t_2-t_1 \\
        &= (\tilde{L}_k,\ldots,\tilde{L}_K,T_k,\ldots,T_K) \mid X>t_2, H_1=\ell_1, T_{j}=t_{j+1}-t_j, \tilde{L}_j=\ell_j \text{ for all } 2 \leq j < k \\
        \intertext{Repeatedly apply Condition~\ref{condition: cond ind cens} at $t_2,\ldots,t_{k-1}$, by the definition of $H_{k-1}$, the display continues as}
        &\overset{d}{=} (\tilde{L}_k,\ldots,\tilde{L}_K,T_k,\ldots,T_K) \mid X>t_k,H_{k-1}=h_{k-1}.
    \end{align*}
    Condition~\ref{condition: positive prob uncensored} is used implicitly to ensure that the event of being uncensored that is conditioned on has a nonzero probability.
    Eq.~\ref{eq: S G at-risk set} follows directly from the identical distributions.
\end{proof}

\begin{proof}[Proof of Theorem~\ref{thm: G identify}]
    We prove by induction on $k$.
    For $k=K$, $\tilde{P}_*(T > \tau \mid T>t_K, \tilde{H}_K) = Q_{*,K}(\tilde{H}_K)$ $P_*$-almost surely by \eqref{eq: S G at-risk set} at $k=K$. Suppose that $\tilde{P}_*(T > \tau \mid T>t_{k+1}, \tilde{H}_{k+1}) = Q_{*,k+1}(\tilde{H}_{k+1})$, then by Lemma~\ref{lemma: at-risk set},
    \begin{align*}
        U_{*,k}(h_k) &= \expect_*[Q_{*,k+1}(H_{k+1}) \mid X>t_{k+1}, H_k=h_k] \\
        &= \expect_*[Q_{*,k+1}(\tilde{H}_{k+1}) \mid T>t_{k+1}, \tilde{H}_k=h_k] \\
        &= \expect_*[\tilde{P}_*(T>\tau \mid T>t_{k+1}, \tilde{H}_{k+1}) \mid T>t_{k+1}, \tilde{H}_k=h_k] \\
        &= \tilde{P}_*(T>\tau \mid T>t_{k+1}, \tilde{H}_k=h_k),
    \end{align*}
    and
    \begin{align*}
        Q_{*,k}(h_k) &= U_{*,k}(h_k) S_{*,k}(t_{k+1} \mid h_k) \\
        &= \tilde{P}_*(T>\tau \mid T>t_{k+1}, \tilde{H}_k=h_k) \tilde{P}_*(T>t_{k+1} \mid T>t_k, \tilde{H}_k=h_k) \\
        &= \tilde{P}_*(T>\tau \mid T>t_k, \tilde{H}_k=h_k).
    \end{align*}
    Thus, by induction, $Q_{*,k}(\tilde{H}_k) = \tilde{P}_*(T > \tau \mid T>t_k, \tilde{H}_k)$ for all $k$, and hence $U_{*,k} = \tilde{P}_*(T>\tau \mid T>t_{k+1},H_k)$ for all $k<K$.
\end{proof}

We next present a stronger version of \eqref{eq: DR transform}, which reveals the analytic expression of the remainder of the transformation $\TS_{k,t}$.

\begin{lemma} \label{lemma: DR transform remainder}
    Suppose that $S_k$ and $G_k$ are two conditional survival functions---that is, both are non-increasing, non-negative, and right-continuous, and both take value 1 at time $t_k$ given any historical covariates---such that $G_k(t_{k+1} \mid H_k)>0$ $P_*$-almost surely.
    Under Conditions~\ref{condition: cond ind cens} and \ref{condition: positive prob uncensored}, for all $t \in (t_k,t_{k+1}]$, with $\bar{R}_{k,t}$ defined in \eqref{eq: DR remainder}, whenever $G_k(t_{k+1} \mid H_k) > 0$,
    $$\expect_* \left[ \TS_{k,t}(S_k,G_k)(X,\Delta \mid H_k) \mid X>t_k, H_k \right] = S_{*,k}(t \mid H_k) + \bar{R}_{k,t}(S_k,G_k \mid H_k).$$
    Moreover, $|\bar{R}_{k,t}(S_k,G_k \mid H_k)| \lesssim \bar{B}_{k,t}(S_k,G_k \mid H_k)$, where $\bar{B}_{k,t}(S_k,G_k \mid H_k)$ is defined as
    \begin{align*}
        \frac{\ind\{G_k(t \mid H_k)>0\}}{G_k(t \mid H_k)} \min \left\{ \sup_{s \in (t_k,t]} |G_k(s \mid H_k) - G_{*,k}(s \mid H_k)|, \sup_{s \in (t_k,t]} |S_k(s \mid H_k) - S_{*,k}(s \mid H_k)| \right\}.
    \end{align*}
\end{lemma}

Recall that, if $S_k(s \mid H_k) = 0$ for some $s \leq t$ so that $S_k(t \mid H_k)=0$, we use the convention that $0 \times \infty := 0$ to define $\bar{R}_{k,t}(S_k,G_k \mid H_k) := 0$.

\begin{proof}
    Without loss of generality, suppose that $t_k=0$. We also drop the subscript $k$ and the conditioning on $\{X>t_k, H_k\}$ from the notation.
    Let $F_*: s \mapsto P_*(X \leq s)$ be the true cumulative distribution function of $X=T \wedge C$. By definition, 
    $$\expect_* \left[ \TS_t(S,G)(X,\Delta) \right] = S(t)-S(t) \expect_* \left[ \frac{\ind(T \leq t, T \leq C)}{S(T) G(T-)} + \int_{(0,X \wedge t]} \frac{S(\intd s)}{S(s) S(s-) G(s-)} \right].$$
    Under Condition~\ref{condition: cond ind cens},
    \begin{equation}
        1-F_*(s) = S_*(s) G_*(s). \label{eq: X T C survival}
    \end{equation}
    By Condition~\ref{condition: cond ind cens}, we have that
    \begin{align*}
        & \expect_*\left[ \frac{\ind(X \leq t) \Delta}{S(X) G(X-)} + \int_{(0,X \wedge t]} \frac{S(\intd s)}{S(s) S(s-) G(s-)} \right] \\
        &= \expect_* \left[ \frac{\ind(T \leq t) \ind(T \leq C)}{S(T) G(T-)} \right] + \int_{(0,\infty)} \int_{(0,x \wedge t]} \frac{S(\intd s)}{S(s) S(s-) G(s-)} F_*(\intd x) \\
        &= -\int_{(0,t]} \frac{G_*(s-)}{S(s) G(s-)} S_*(\intd s) + \int_{(0,t]} \frac{S(\intd s)}{S(s) S(s-) G(s-)} \int_{[s,\infty)} F_*(\intd x) & \text{(Fubini's Theorem)} \\
        &= -\int_{(0,t]} \frac{G_*(s-)}{S(s) G(s-)} S_*(\intd s) + \int_{(0,t]} \frac{S(\intd s)}{S(s) S(s-) G(s-)} S_*(s-) G_*(s-) & \text{(by \eqref{eq: X T C survival})} \\
        &= - \int_{(0,t]} \frac{1}{S(s)} S_*(\intd s) + \int_{(0,t]} \frac{S_*(s-)}{S(s) S(s-)} S(\intd s) \\
        &\quad- \int_{(0,t]} \frac{G_*(s-) - G(s-)}{S(s) G(s-)} S_*(\intd s) + \int_{(0,t]} \frac{S_*(s-) \{G_*(s-) - G(s-)\}}{S(s) S(s-) G(s-)} S(\intd s) \\
        &= \int_{(0,t]} \left\{ \frac{(S-S_*)(\intd s)}{S(s)} - \frac{S(s-)-S_*(s-)}{S(s) S(s-)} S(\intd s) \right\} \\
        &\quad+ \int_{(0,t]} \frac{G_*(s-)-G(s-)}{G(s-)} \left\{ \frac{(S-S_*)(\intd s)}{S(s)} - \frac{S(s-)-S_*(s-)}{S(s) S(s-)} S(\intd s) \right\} \\
        &= \int_{(0,t]} \left\{ \frac{(S-S_*)(\intd s)}{S(s)} - [S(s-)-S_*(s-)] \left(\frac{1}{S}\right)(\intd s) \right\} \\
        &\quad+ \int_{(0,t]} \frac{G_*(s-)-G(s-)}{G(s-)} \left\{ \frac{(S-S_*)(\intd s)}{S(s)} - [S(s-)-S_*(s-)] \left(\frac{1}{S}\right)(\intd s) \right\} \\
        &= \frac{S(t)-S_*(t)}{S(t)} - \frac{S(0)-S_*(0)}{S(0)} + \int_{(0,t]} \frac{G_*(s-)-G(s-)}{G(s-)} \left( \frac{S-S_*}{S} \right)(\intd s) & \text{(integration by parts)} \\
        &= \frac{S(t)-S_*(t)}{S(t)} + \int_{(0,t]} \frac{G_*(s-)-G(s-)}{G(s-)} \left( \frac{S-S_*}{S} \right)(\intd s). & \text{($S(0)=S_*(0)=1$)}
    \end{align*}
    Using the full notation again, without assuming $t_k=0$, the above equality implies that
    \begin{align*}
        &\expect_* \left[ \TS_{k,t}(S_k,G_k)(X,\Delta \mid H_k) \mid X>t_k, H_k \right] - S_{*,k}(t \mid H_k) \\
        &= S_k(t \mid H_k) - S_{*,k}(t \mid H_k) - \{ S_k(t \mid H_k)-S_{*,k}(t \mid H_k) \} \\
        &\quad- S_k(t \mid H_k) \int_{(t_k,t]} \frac{G_k(s- \mid H_k)-G_{*,k}(s- \mid H_k)}{G_k(s- \mid H_k)} \left( \frac{S_k-S_{*,k}}{S_k} \right)(\intd s \mid H_k) \\
        &= \bar{R}_{k,t}(S_k,G_k \mid H_k).
    \end{align*}

    We next prove the bound on $\bar{R}_{k,t}(S_k,G_k \mid H_k)$. For $H_k$ in the zero-probability set such that $G_k(t_{k+1} \mid H_k) = 0$, the bound is trivial since $\bar{R}_{k,t}(S_k,G_k \mid H_k)=\bar{B}_{k,t}(S_k,G_k \mid H_k)=0$.
    Otherwise, by the fact that $t \mapsto S_k(t \mid h_k)$ and $t \mapsto G_k(t \mid h_k)$ are non-increasing and bounded by one,
    \begin{align}
        |\bar{R}_{k,t}(S_k,G_k \mid H_k)| &\leq |S_k(t \mid H_k)| \int_{(t_k,t]} \left| \frac{G_k(s- \mid H_k)-G_{*,k}(s- \mid H_k)}{G_k(s- \mid H_k)} \right| \left| \left( \frac{S_k-S_{*,k}}{S_k} \right)(\intd s \mid H_k) \right| \nonumber \\
        &\lesssim \frac{1}{G_k(t \mid H_k)} \int_{(t_k,t]} \left| G_k(s- \mid H_k)-G_{*,k}(s- \mid H_k) \right| \left| (S_k-S_{*,k})(\intd s \mid H_k) \right| \nonumber \\
        &\lesssim \frac{1}{G_k(t \mid H_k)} \sup_{s \in (t_k,t]} |G_k(s \mid H_k) - G_{*,k}(s \mid H_k)|. \label{eq: bound1 on remainder}
    \end{align}
    By integration by parts,
    \begin{align*}
        \bar{R}_{k,t}(S_k,G_k \mid H_k) &= -\frac{ \{G_k(t \mid H_k)-G_{*,k}(t \mid H_k)\} \{S_k(t \mid H_k)-S_{*,k}(t \mid H_k)\}}{G_k(t \mid H_k)} \\
        &\quad+ \frac{ \{G_k(t_k \mid H_k)-G_{*,k}(t_k \mid H_k)\} \{S_k(t_k \mid H_k)-S_{*,k}(t_k \mid H_k)\}}{G_k(t_k \mid H_k)} \\
        &\quad+ S_k(t \mid H_k) \int_{(t_k,t]} \frac{S_k(s- \mid H_k)-S_{*,k}(s- \mid H_k)}{S_k(s- \mid H_k)} \left( \frac{G_k-G_{*,k}}{G_k} \right)(\intd s \mid H_k) \\
        &= -\frac{ \{G_k(t \mid H_k)-G_{*,k}(t \mid H_k)\} \{S_k(t \mid H_k)-S_{*,k}(t \mid H_k)\}}{G_k(t \mid H_k)} \\
        &\quad+ S_k(t \mid H_k) \int_{(t_k,t]} \frac{S_k(s- \mid H_k)-S_{*,k}(s- \mid H_k)}{S_k(s- \mid H_k)} \left( \frac{G_k-G_{*,k}}{G_k} \right)(\intd s \mid H_k).
    \end{align*}
    By a similar argument as above,
    \begin{align}
        |\bar{R}_{k,t}(S_k,G_k \mid H_k)| &\lesssim \frac{ |G_k(t \mid H_k)-G_{*,k}(t \mid H_k)| |S_k(t \mid H_k)-S_{*,k}(t \mid H_k)|}{G_k(t \mid H_k)} \nonumber \\
        &\quad+ \frac{1}{G_k(t \mid H_k)} \sup_{s \in (t_k,t]} |S_k(s \mid H_k) - S_{*,k}(s \mid H_k)| \nonumber \\
        &\lesssim \frac{1}{G_k(t \mid H_k)} \sup_{s \in (t_k,t]} |S_k(s \mid H_k) - S_{*,k}(s \mid H_k)|. \label{eq: bound2 on remainder}
    \end{align}
    Since both \eqref{eq: bound1 on remainder} and \eqref{eq: bound2 on remainder} hold, $|\bar{R}_{k,t}(S_k,G_k \mid H_k)| \lesssim \bar{B}_{k,t}(S_k,G_k \mid H_k)$.
\end{proof}

We next present a lemma that is useful to handle a misalignment in the domain of $U_{*,k}$. This function is only well defined for $h_k \in \mathcal{H}_k$ such that $P_*(X>t_{k+1} \mid X>t_k, H_k=h_k) > 0$ and might not be well defined for all $h_k$ in the entire support of $H_k \mid X>t_k$.
It is thus necessary to multiply by $S_{*,k}(t_{k+1} \mid h_k)$ when studying the distance induced by $\| \cdot\|_k$ between, e.g., $U_k: \mathcal{H}_k \mapsto \real$ and $U_{*,k}$, so that their difference is defined to be zero for $h_k$ such that $P_*(X>t_{k+1} \mid X>t_k, H_k=h_k) = 0$.
The following lemma connects this distance to $\| \cdot \|_{k+1}$, where $U_{*,k}$ is well defined.

\begin{lemma} \label{lemma: k k+1 norms}
    Let $f: \mathcal{H}_k \to \real$ be any measurable function. Under Conditions~\ref{condition: cond ind cens} and \ref{condition: consistency of strong positivity}, for all $k = \varsigma,\ldots,K-1$, $\| S_{*,k}(t_{k+1} \mid \cdot) f(\cdot) \|_{k} \lesssim \| f \|_{k+1}$.
\end{lemma}

\begin{proof}
    Let $p_{H_k \mid X>t_k}$ and $p_{H_k \mid X>t_{k+1}}$ denote the conditional densities of $H_k \mid X>t_k$ and $H_k \mid X>t_{k+1}$, respectively, with respect to a dominating measure $\mu$. By Bayes' Theorem,
    $$\frac{p_{H_k \mid X>t_k}}{p_{H_k \mid X>t_{k+1}}}(h_k) = \frac{P_*(X>t_{k+1} \mid X>t_k)}{P_*(X>t_{k+1} \mid X>t_k, H_k=h_k)} = \frac{P_*(X>t_{k+1} \mid X>t_k)}{S_{*,k}(t_{k+1} \mid h_k) G_{*,k}(t_{k+1} \mid h_k)}.$$
    Therefore,
    \begin{align*}
        \| S_{*,k}(t_{k+1} \mid \cdot) f(\cdot) \|_{k}^2 &= \int S_{*,k}(t_{k+1} \mid h_k)^2 f(h_k)^2 p_{H_k \mid X>t_k}(h_k) \intd \mu(h_k) \\
        &= \int S_{*,k}(t_{k+1} \mid h_k)^2 f(h_k)^2 p_{H_k \mid X>t_{k+1}}(h_k) \frac{P_*(X>t_{k+1} \mid X>t_k)}{S_{*,k}(t_{k+1} \mid h_k) G_{*,k}(t_{k+1} \mid h_k)} \intd \mu(h_k) \\
        &= \int S_{*,k}(t_{k+1} \mid h_k) f(h_k)^2 p_{H_k \mid X>t_{k+1}}(h_k) \frac{P_*(X>t_{k+1} \mid X>t_k)}{G_{*,k}(t_{k+1} \mid h_k)} \intd \mu(h_k) \\
        &\lesssim \int f(h_k)^2 p_{H_k \mid X>t_{k+1}}(h_k) \intd \mu(h_k) = \| f \|_{k+1}^2,
    \end{align*}
    where the last step follows from Condition~\ref{condition: consistency of strong positivity} and that $S_{*,k} \leq 1$, $P_*(X>t_{k+1} \mid X>t_k) \leq 1$.
    Thus, Lemma~\ref{lemma: k k+1 norms} holds.
\end{proof}

We next present a stronger version of Theorem~\ref{thm: SDR identify}, which reveals the analytic expression of the remainder of the transformation $\T_{k}$.
\begin{lemma} \label{lemma: SDR transform remainder}
    Suppose that, for all $k=1,\ldots,K$, $S_k$ and $G_k$ are conditional survival functions such that $G_k(t_{k+1} \mid H_k)>0$ $P_*$-almost surely; $U_k$ is any square-integrable function with $U_K=1$.
    Define recursively
    \begin{align}
        &R_{K}(\{U_K,S_K,G_K\} \mid H_K) := \bar{R}_{K,\tau}(S_K,G_K \mid H_K), \nonumber \\
        &R_{k}(\{U_j, S_j, G_j\}_{j=k}^K \mid H_k) \nonumber \\
        &:= \frac{S_{*,k}(t_{k+1} \mid H_k) G_{*,k}(t_{k+1} \mid H_k)}{G_k(t_{k+1} \mid H_k)} \expect_*[R_{k+1}(\{U_j, S_j, G_j\}_{j=k+1}^K \mid H_{k+1}) \mid X>t_{k+1}, H_k] \nonumber \\
        &\quad+ \frac{S_{*,k}(t_{k+1} \mid H_k)}{G_k(t_{k+1} \mid H_k)} \{ G_k(t_{k+1} \mid H_k) - G_{*,k}(t_{k+1} \mid H_k) \} \{ U_k(H_k) - U_{*,k}(H_k) \} \nonumber \\
        &\quad+ U_k(H_k) \bar{R}_{k,t_{k+1}}(S_k,G_k \mid H_k) \qquad \text{ for } k=K-1,\ldots,1 \label{eq: SDR remainder recursive}
    \end{align}
    whenever $G_k(t_{k+1} \mid H_k) > 0$ and zero otherwise.
    Under Conditions~\ref{condition: cond ind cens} and \ref{condition: positive prob uncensored}, whenever $G_k(t_{k+1} \mid H_k) > 0$,
    \begin{equation}
        \expect_*[\T_{k}( \{U_j, S_j, G_j\}_{j=k}^K)(O) \mid X>t_k, H_k] = Q_{*,k}(H_k) + R_{k}(\{U_j, S_j, G_j\}_{j=k}^K \mid H_k). \label{eq: SDR remainder}
    \end{equation}
    Moreover, if there exist constants $\delta_k>0$ such that, for all $k$ and $P_*$-almost every $h_k$, $G_k(t_{k+1} \mid H_k=h_k) \geq \delta_k >0$ and $U_k(h_k)$ is bounded, then $\| R_{k}(\{U_j, S_j, G_j\}_{j=k}^K \mid \cdot) \|_k \lesssim B_{k}(\{U_j, S_j, G_j\}_{j=k}^K)$ where
    \begin{align*}
        B_{k}(\{U_j, S_j, G_j\}_{j=k}^K) &:= \sum_{j=k}^K \left\{ \prod_{\ell=k}^{j-1} \frac{1}{\delta_\ell} \right\} \left\{ \|G_j - G_{*,j}\|_{(t_j,t_{j+1}],j} \| U_j - U_{*,j} \|_{j+1} + \| \bar{B}_{j,t_{j+1}}(S_j,G_j \mid \cdot) \|_j \right\}.
    \end{align*}
    If, for all $j$, either (i) $S_j=S_{*,j}$ for all $t \in (t_j,t_{j+1}]$ and $U_j=U_{*,j}$, or (ii) $G_j=G_{*,j}$ for all $t \in (t_j,t_{j+1}]$, then $B_{k}(\{U_j, S_j, G_j\}_{j=k}^K) = 0$ and the assumptions that $G_k(t_{k+1} \mid H_k)$ is bounded away from 0 and that $U_k$ is bounded almost surely can be dropped.
\end{lemma}

\begin{proof}
    We prove by induction from $k=K$ to $k=1$. First consider the case where $G_k(t_{k+1} \mid H_k) \geq \delta_k$ and $U_K$ is bounded almost surely.
    Since $\T_{K}(\{U_K,S_K,G_K\})(O) = \TS_{K,t_{K+1}}(S_K,G_K)(X,\Delta \mid H_K)$ and $B_{K}(\{U_K, S_K, G_K\}) = \| \bar{B}_{K,t_{k+1}}(S_K,G_K \mid \cdot) \|_K$, Lemma~\ref{lemma: SDR transform remainder} holds for $k=K$ by Lemma~\ref{lemma: DR transform remainder}.
    
    Suppose that Lemma~\ref{lemma: SDR transform remainder} holds for $k+1$ and consider $k$.
    When $H_k$ is in the one-probability set satisfying the boundedness condition,
    by \eqref{eq: SDR transform recursive}, Condition~\ref{condition: cond ind cens} and the induction hypothesis,
    \begin{align*}
        &\expect_*[\T_{k}( \{U_j, S_j, G_j\}_{j=k}^K)(O) \mid X>t_k,H_k] \\
        &= \frac{S_{*,k}(t_{k+1} \mid H_k) G_{*,k}(t_{k+1} \mid H_{k+1})}{G_k(t_{k+1} \mid H_k)} \{ \expect_*[Q_{*,k+1}(H_{k+1}) \mid X>t_{k+1},H_k] \\
        &\qquad\qquad\qquad\qquad+ \expect_*[R_{k+1}(\{U_j, S_j, G_j\}_{j=k+1}^K \mid H_{k+1}) \mid X>t_{k+1},H_k] - U_k(H_k) \} \\
        &\qquad\qquad+ U_k(H_k) \expect_*[ \TS_{k,t_{k+1}}(S_k,G_k)(X,\Delta \mid H_k) \mid X>t_k,H_k] \\
        &= \frac{S_{*,k}(t_{k+1} \mid H_k) G_{*,k}(t_{k+1} \mid H_k)}{G_k(t_{k+1} \mid H_k)} \\
        &\qquad\qquad\times \left\{ U_{*,k}(H_k) - U_k(H_k) + \expect_*[R_{k+1}(\{U_j, S_j, G_j\}_{j=k+1}^K \mid H_{k+1}) \mid X>t_{k+1},H_k] \right\} \\
        &\qquad+ U_k(H_k) \{ S_{*,k}(t_{k+1}) + \bar{R}_{k,t_{k+1}}(S_k,G_k \mid H_k) \} \qquad \text{(Lemma~\ref{lemma: DR transform remainder})} \\
        &= S_{*,k}(t_{k+1} \mid H_k) G_{*,k}(t_{k+1} \mid H_k) \left\{ \frac{1}{G_{*,k}(t_{k+1} \mid H_k)} + \frac{1}{G_k(t_{k+1} \mid H_k)} - \frac{1}{G_{*,k}(t_{k+1} \mid H_k)} \right\} \\
        &\qquad\qquad\times \left\{ U_{*,k}(H_k) - U_k(H_k) + \expect_*[R_{k+1}(\{U_j, S_j, G_j\}_{j=k+1}^K \mid H_{k+1}) \mid X>t_{k+1},H_k] \right\} \\
        &\qquad+ U_k(H_k) \{ S_{*,k}(t_{k+1}) + \bar{R}_{k,t_{k+1}}(S_k,G_k \mid H_k) \} \\
        &= Q_{*,k}(H_k) + R_{k}(\{U_j, S_j, G_j\}_{j=k}^K \mid H_k).
    \end{align*}
    Moreover, note that
    \begin{align*}
        B_{k}(\{U_j, S_j, G_j\}_{j=k}^K) &= \frac{1}{\delta_k} B_{k+1}(\{U_j, S_j, G_j\}_{j=k+1}^K) \\
        &\quad+ \|G_k - G_{*,k}\|_{(t_k,t_{k+1}],k} \| U_k - U_{*,k} \|_{k+1} + \| \bar{B}_{k,t_{k+1}}(S_k,G_k \mid \cdot) \|_k.
    \end{align*}
    By induction hypothesis, $\| R_{k+1}(\{U_j, S_j, G_j\}_{j=k+1}^K \mid \cdot) \|_{k+1} \lesssim B_{k+1}(\{U_j, S_j, G_j\}_{j=k+1}^K)$.
    By Lemma~\ref{lemma: DR transform remainder}, $\| \bar{R}_{k,t_{k+1}}(S_k,G_k \mid \cdot) \|_k \lesssim \| \bar{B}_{k,t_{k+1}}(S_k,G_k \mid \cdot) \|_k$.
    Thus, $\| R_{k}(\{U_j, S_j, G_j\}_{j=k}^K \mid \cdot) \|_k \lesssim B_{k}(\{U_j, S_j, G_j\}_{j=k}^K)$ by Lemma~\ref{lemma: k k+1 norms}, so Lemma~\ref{lemma: SDR transform remainder} holds for $k$. By induction, Lemma~\ref{lemma: SDR transform remainder} holds for all $k$.

    If, for all $k$, either (i) $U_k=U_{*,k}$ and $S_k=S_{*,k}$, or (ii) $G_k=G_{*,k}$, then $R_{K}(\{U_K, S_K, G_K\} \mid H_K)$ and $B_{K}(\{U_K, S_K, G_K\})$ are both zero by Lemma~\ref{lemma: DR transform remainder}. By a similar induction argument, $R_{k}(\{U_j, S_j, G_j\}_{j=k}^K \mid H_k)$ and $B_{k}(\{U_j, S_j, G_j\}_{j=k}^K)$ are both zero for all $k$ and thus Lemma~\ref{lemma: SDR transform remainder} holds for all $k$.
\end{proof}

\subsection{Bounds} \label{sec: proof bound}

\begin{proof}[Proof of Theorem~\ref{thm: bound}]
    The results \eqref{eq: SDR point bound} for the SDR oracle estimator follow immediately from Lemmas~\ref{lemma: SDR transform remainder} and \ref{lemma: k k+1 norms} by restricting to the high-probability event that $\hat{U}_k^\SDR$ is $P_*$-almost surely bounded and $\hat{G}_k$ is bounded away from zero, and substituting in $(\hat{U}_{j}^\SDR,\hat{S}_j,\hat{G}_j)$ for $(U_j,S_j,G_j)$.
    Since $U_{*,k}$ and $S_{*,k}$ are bounded, the results \eqref{eq: IPCW point bound} for the IPCW oracle estimator also follow from Lemma~\ref{lemma: SDR transform remainder} by the substitution described above Corollary~\ref{coro: IPCW identify}.
    We next prove the bounds \eqref{eq: G point bound} for the G-computation estimator. Since
    \begin{align*}
        & Y^\G_{k} - Q_{*,k}(H_k) = \{ \hat{S}_{k}(t_{k+1} \mid H_k) - S_{*,k}(t_{k+1} \mid H_k) \} \hat{U}_{k}^\G(H_k) + S_{*,k}(t_{k+1} \mid H_k) \{\hat{U}_{k}^\G(H_k) - U_{*,k}(H_k)\},
    \end{align*}
    when $\hat{U}_{k}^\G$ is $P_*$-almost surely bounded, it holds that
    $$\left| Y^\G_{k} - Q_{*,k}(H_k) \right| \lesssim |\hat{S}_{k}(t_{k+1} \mid H_k) - S_{*,k}(t_{k+1} \mid H_k)| + |S_{*,k}(t_{k+1} \mid H_k) \{\hat{U}_{k}^\G(H_k) - U_{*,k}(H_k)\} |$$
    $P_*$-almost surely, and \eqref{eq: G point bound} follows from Lemma~\ref{lemma: k k+1 norms}.
\end{proof}

\begin{proof}[Proof of Corollary~\ref{coro: SDR}]
    We use an induction argument similar to the proof of Lemma~\ref{lemma: SDR transform remainder}. We first show that, for every $j = \varsigma,\ldots,K$, $\| \hat{U}_{j}^\SDR-U_{*,j} \|_{j+1} = \smallo_p(1)$ if $\| \hat{U}_{j}^\SDR-\bar{U}_{j}^\SDR \|_{j+1} = \smallo_p(1)$.
    Since $\hat{U}_{K}^\SDR=\bar{U}_{K}^\SDR=U_{*,K}=1$, the claim holds trivially for $K$.

    Consider a general $k = \varsigma,\ldots,K-1$ and suppose that the claim holds for $k+1,\ldots,K$.
    If $\| \hat{U}_{k}^\SDR-\bar{U}_{k}^\SDR \|_{k+1} = \smallo_p(1)$, then, by \eqref{eq: SDR point bound}, the induction hypothesis, and the assumption of consistent estimation of $(S_{*,j},\bar{U}_{j}^\SDR)$ or $G_{*,j}$ for each $j = k+1,\ldots,K$, it holds that $(S_{*,j},U_{*,j})$ or $G_{*,j}$ is estimated consistently for each $j = k+1,\ldots,K$, and thus $\| \hat{U}_{k}^\SDR-U_{*,k} \|_{k+1} \leq \| \hat{U}_{k}^\SDR-\bar{U}_{k}^\SDR \|_{k+1} + \| \bar{U}_{k}^\SDR-U_{*,k} \|_{k+1} = \smallo_p(1)$, so the claim holds for $k$.
    Thus, by induction, the claim holds for all $j=\varsigma,\ldots,K$.
    
    Using a similar argument, by \eqref{eq: SDR point bound}, $\| \hat{Q}_{\varsigma}^\SDR-\bar{Q}_{\varsigma}^\SDR \|_\varsigma = \smallo_p(1)$, and the assumption of consistent estimation of $(S_{*,j},\bar{U}_{j}^\SDR)$ or $G_{*,j}$ for each $j = \varsigma,\ldots,K$, it holds that $\| \hat{Q}_{\varsigma}^\SDR-Q_{*,\varsigma} \|_\varsigma \leq \| \hat{Q}_{\varsigma}^\SDR-\bar{Q}_{\varsigma}^\SDR \|_\varsigma + \| \bar{Q}_{\varsigma}^\SDR-Q_{*,\varsigma} \|_\varsigma = \smallo_p(1)$.
\end{proof}

\subsection{Asymptotic efficiency} \label{sec: proof efficiency}

For any distribution $P$ and a $P$-integrable function $f$, we use $P f$ or $P f(O)$ to denote $\int f \intd P = \int f(o) P(\intd o)$.
Note that, if $f$ is random, $P_* f$ is also random by definition. 
Let $P_n$ denote the empirical distribution.
For any distribution $P$ on the observed data, we replace the subscript $*$ with $P$ to denote functions and quantities of $P$ rather than $P_*$. For example, we use $\alpha_{P}$ to denote the marginal survival probability $\alpha_{*}$ with $P_*$ replaced by $P$.
We also use $\expect_P$ to denote expectation with respect to $P$.

\begin{proof}[Proof of Theorem~\ref{thm: efficiency}: Efficient influence function]
    In this proof, we first show that $\T_{k}( \{U_j, S_j, G_j\}_{j=k}^K) - Q_{*,k}$ is the conditional influence function for estimating $Q_{*,k}$, and then show that $D(\{U_{*,j}, S_{*,j}, G_{*,j}\}_{j=1}^K, \allowbreak \alpha_*)$ is the efficient influence function for estimating $\alpha_*$ under a nonparametric model. With $P$ denoting a generic distribution, $Q_{P,k}$ denoting its implied conditional survival function at $\tau$ given survival through $t_k$, and $\alpha_P$ denoting its implied marginal survival probability at $\tau$, it suffices to derive the pathwise derivatives $P \mapsto Q_{P,k}$ and $P \mapsto \alpha_P$.
    We use an induction argument similar to Lemma~\ref{lemma: SDR transform remainder} and Corollary~\ref{coro: SDR}.
    Note that $\sigma_*^2 < \infty$ by Condition~\ref{condition: consistency of strong positivity}.

    It will be convenient to consider time windows $(t_k,t_{k+1}]$ between adjacent visit times $t_k$ and $t_{k+1}$ one by one.
    Recall the definition of $T_k$, $C_k$, etc. in Section~\ref{sec: proof identify}.
    For any two indices $j_1$ and $j_2$, let subscript $[j_1:j_2]$ denote the collection from index $j_1$ to $j_2$ if $j_1 \leq j_2$---for example, $x_{[1:k]} := (x_1,\ldots,x_k)$---and denote the empty set if $j_1 > j_2$.
    Under distribution $P_*$, let $P_{*,L_k}$ denote the distribution of $L_k$ conditional on $(H_{k-1},X_{[1:(k-1)]},\Delta_{[1:(k-1)]})$ (recall that we define $H_k \equiv 0$ if $X \leq t_k$), and $P_{*,X_{k},\Delta_{k}}$ denote the distribution of $(X_{k},\Delta_{k})$ conditional on $(H_k,X_{[1:(k-1)]},\Delta_{[1:(k-1)]})$.

    Let $L^2_0(P_*)$ denote the set of score functions $f: \mathcal{O} \rightarrow \real$ such that $\expect_*[f(O)]=0$ and $\expect_*[f(O)^2]<\infty$, which is the tangent set at $P_*$ under a nonparametric model. Let $\mathcal{V}$ denote the collection of functions in $L^2_0(P_*)$ with range contained in $[-1,1]$. We note that the $L^2_0(P_*)$-closure of the linear span of $\mathcal{V}$ is indeed $L^2_0(P_*)$. For a score $V \in \mathcal{V}$, define the following functions recursively from $k=1$ to $k=K$ (recall $\sum_a^b \cdots = 0$ whenever $a>b$):
    \begin{align*}
        &V_{L_k}: (\ell_k \mid h_{k-1},x_{[1:(k-1)]},\delta_{[1:(k-1)]}) \\
        &\qquad\qquad\mapsto \expect_*[V(O) \mid L_k=\ell_k, H_{k-1}=h_{k-1}, X_{[1:(k-1)]}=x_{[1:(k-1)]}, \Delta_{[1:(k-1)]}=\delta_{[1:(k-1)]}] \\
        &\qquad\qquad\qquad- \sum_{j=1}^{k-1} V_{L_j}(\ell_j \mid h_{j-1}, x_{[1:(j-1)]}, \delta_{[1:(j-1)]}) - \sum_{j=1}^{k-1} V_{X_{j},\Delta_{j}}(x_{j},\delta_{j} \mid h_j,x_{[1:(j-1)]},\delta_{[1:(j-1)]}), \\
        &V_{X_{k},\Delta_{k}}: (x_{k},\delta_{k} \mid h_k,x_{[1:(k-1)]},\delta_{[1:(k-1)]}) \\
        &\qquad\qquad\mapsto \expect_*[V(O) \mid X_{k}=x_{k}, \Delta_{k}=\delta_{k}, H_k=h_k, X_{[1:(k-1)]}=x_{[1:(k-1)]},\Delta_{[1:(k-1)]}=\delta_{[1:(k-1)]}] \\
        &\qquad\qquad\qquad- \sum_{j=1}^k V_{L_j}(\ell_j \mid h_{j-1}, x_{[1:(j-1)]}, \delta_{[1:(j-1)]}) - \sum_{j=1}^{k-1} V_{X_{j},\Delta_{j}}(x_{j},\delta_{j} \mid h_j,x_{[1:(j-1)]},\delta_{[1:(j-1)]}).
    \end{align*}
    We note that $\expect_*[V_{L_k}(L_k \mid H_{k-1},X_{[1:k]},\Delta_{[1:k]}) \mid H_{k-1},X_{[1:k]},\Delta_{[1:k]}] = 0$, $\expect_*[V_{X_{k},\Delta_{k}}(X_{k},\Delta_{k} \mid H_k,X_{[1:(k-1)]},\Delta_{[1:(k-1)]}) \mid H_k,X_{[1:(k-1)]},\Delta_{[1:(k-1)]}] = 0$, $P_* V_{L_k} V_{X_j,\Delta_j} = 0$ for all $k,j$, $P_* V_{L_k} V_{L_j} = 0$ for all $k \neq j$, $P_* V_{X_k,\Delta_k} V_{X_j,\Delta_j} = 0$ for all $k \neq j$, and $V = \sum_{k=1}^K (V_{L_k} + V_{X_{k},\Delta_{k}})$.

    Let $V \in \mathcal{V}$ be fixed. For all $\epsilon$ in a sufficiently small neighborhood $\mathcal{N}$ of zero, we define $P_\epsilon$ via its Radon-Nikodym derivative with respect to $P_*$:
    \begin{align*}
        \frac{\intd P_\epsilon}{\intd P_*}: o &\mapsto \prod_{k=1}^K \left\{ 1+\epsilon V_{L_k}(\ell_k \mid h_{k-1},x_{[1:(k-1)]},\delta_{[1:(k-1)]}) \right\} \left\{ 1+\epsilon V_{X_{k},\Delta_{k}}(x_{k},\delta_{k} \mid h_k,x_{[1:(k-1)]},\delta_{[1:(k-1)]}) \right\}.
    \end{align*}
    It is straightforward to verify that the score function of the parametric submodel $\{ P_\epsilon: \epsilon \in \mathcal{N}\}$ at $\epsilon=0$ is $V$.
    We also define components $P_{\epsilon,L_k}$ and $P_{\epsilon,X_{k},\Delta_{k}}$ of $P_\epsilon$ similarly to $P_*$. In addition, similarly to $P_*$, we use the subscript $\epsilon$ for quantities/functions of $P_\epsilon$.
    It is also straightforward to verify that
    \begin{align*}
         \frac{\intd P_{\epsilon,L_k}}{\intd P_{*,L_k}} (\ell_k \mid h_{k-1},x_{[1:(k-1)]},\delta_{[1:(k-1)]}) &= 1+\epsilon V_{L_k}(\ell_k \mid h_{k-1},x_{[1:(k-1)]},\delta_{[1:(k-1)]}), \\
         \frac{\intd P_{\epsilon,X_{k},\Delta_{k}}}{\intd P_{*,X_{k},\Delta_{k}}} (x_{k},\delta_{k} \mid h_k,x_{[1:(k-1)]},\delta_{[1:(k-1)]}) &= 1+\epsilon V_{X_{k},\Delta_{k}}(x_{k},\delta_{k} \mid h_k,x_{[1:(k-1)]},\delta_{[1:(k-1)]}).
    \end{align*}

    We first show the pathwise derivative of $P \mapsto S_{P,k}$ for all $k$. 
    Let $x^\dagger := (t_2-t_1,t_3-t_2,\ldots,t_K-t_{K-1})$.
    \protect\citetsupp{Gill1994,Reid1981} showed the pathwise differentiability of the marginal survival function under independent censoring.
    We apply this result conditional on $X>t_k$ and $H_k$ to obtain that, under Conditions~\ref{condition: cond ind cens} and \ref{condition: consistency of strong positivity},
    \begin{align}
        &\left. \frac{\partial S_{\epsilon,k}(t_{k+1} \mid h_k)}{\partial \epsilon} \right|_{\epsilon=0} \nonumber\\
        &= \int \left\{ \TS_{k,t_{k+1}}(S_{*,k},G_{*,k})(x,\delta \mid h_k) - S_{*,k}(t_{k+1} \mid h_k) \right\} V_{X_{k},\Delta_{k}}(x_{k},\delta_{k} \mid h_k,x_{[1:(k-1)]}=x^\dagger_{[1:(k-1)]},\delta_{[1:(k-1)]}=0) \nonumber \\
        &\qquad\qquad P_{*,X_k,\Delta_k}(\intd x_k,\intd \delta_k \mid h_k,x_{[1:(k-1)]}=x^\dagger_{[1:(k-1)]},\delta_{[1:(k-1)]}=0) \nonumber \\
        &= \expect_*[ \left\{ \TS_{k,t_{k+1}}(S_{*,k},G_{*,k})(X,\Delta \mid h_k) - S_{*,k}(t_{k+1} \mid h_k) \right\} V(O) \mid X>t_k,H_k=h_k ], \label{eq: Sk IF}
    \end{align}
    where the last step follows from the facts that $\TS_{k,t_{k+1}}(S_{*,k},G_{*,k})(x,\delta \mid h_k) - S_{*,k}(t_{k+1} \mid h_k)$ only depends on $(x_{[1:k]},\delta_{[1:k]},h_k)$ and that it has mean zero given $X>t_k,H_k=h_k$ (Lemma~\ref{lemma: DR transform remainder}).

    We next prove by induction that, for all $k=1,\ldots,K$,
    \begin{align}
        \left. \frac{\partial Q_{\epsilon,k}(h_k)}{\partial \epsilon} \right|_{\epsilon=0} &= \expect_* \left[ \left\{ \T_{k}( \{U_j, S_j, G_j\}_{j=k}^K)(O) - Q_{*,k}(h_k) \right\} V(O) \mid X>t_k,H_k=h_k \right] \label{eq: Qk IF}
    \end{align}
    Since $V \in \mathcal{V}$ is arbitrary, the closure of $\mathcal{V}$ is the tangent set $L_0^2(P_*)$, and $\expect_*[\T_{k}( \{U_j, S_j, G_j\}_{j=k}^K)(O) - Q_{*,k}(h_k) \mid X>t_k,H_k=h_k]=0$ (by Theorem~\ref{thm: SDR identify}), \eqref{eq: Qk IF} implies that $\T_{k}( \{U_j, S_j, G_j\}_{j=k}^K) - Q_{*,k}$ is a conditional influence function for estimating $Q_{*,k}$ \protect\citepsupp{Chernozhukov2024}, so the SDR transformation in \eqref{eq: SDR transform} is a one-step correction based on this conditional influence function.
    
    For $k=K$, $Q_{\epsilon,k}(h_k) = S_{\epsilon,k}(\tau \mid h_k)$ and $\T_{k}( \{U_{*,j}, S_{*,j}, G_{*,j}\}_{j=k}^K)(O) = \TS_{k,t_{k+1}}(S_{*,k},G_{*,k})(X,\Delta \mid H_k)$, so \eqref{eq: Qk IF} follows from \eqref{eq: Sk IF}. Suppose that \eqref{eq: Qk IF} holds for $k+1$. 
    Note that we can write $h_{k+1}$ as $(\ell_{k+1},h_k)$.
    By the recursive definition in Theorem~\ref{thm: G identify},
    \begin{align*}
        &\left. \frac{\partial Q_{\epsilon,k}(h_k)}{\partial \epsilon} \right|_{\epsilon=0} \\
        &= \left. \frac{\partial}{\partial \epsilon} \expect_{P_\epsilon} \left[ Q_{\epsilon,k+1}(H_{k+1}) \mid X>t_{k+1}, H_k=h_k \right] S_{\epsilon,k}(t_{k+1} \mid h_k) \right|_{\epsilon=0} \\
        &= \frac{\partial}{\partial \epsilon} S_{\epsilon,k}(t_{k+1} \mid h_k) \int Q_{\epsilon,k+1}(\ell_{k+1},h_k) \{ 1+\epsilon V_{L_{k+1}}(\ell_{k+1} \mid h_k,x_{[1:k]}=x^\dagger_{[1:k]},\delta_{[1:k]}=0) \} \\
        &\qquad\qquad P_{*,L_{k+1}}(\intd \ell_{k+1} \mid h_k,x_{[1:k]}=x^\dagger_{[1:k]},\delta_{[1:k]}=0) \Bigg|_{\epsilon=0} \\
        &= U_{*,k}(h_k) \int \left\{ \TS_{k,t_{k+1}}(S_{*,k},G_{*,k})(x,\delta \mid h_k) - S_{*,k}(t_{k+1} \mid h_k) \right\} \\
        &\qquad\qquad V_{X_{k},\Delta_{k}}(x_{k},\delta_{k} \mid h_k,x_{[1:(k-1)]}=x^\dagger_{[1:(k-1)]},\delta_{[1:(k-1)]}=0) \\
        &\qquad\qquad P_{*,X_k,\Delta_k}(\intd x_k,\intd \delta_k \mid h_k,x_{[1:(k-1)]}=x^\dagger_{[1:(k-1)]},\delta_{[1:(k-1)]}=0) \qquad \text{(Eq.~\ref{eq: Sk IF})} \\
        &\quad+ S_{*,k}(t_{k+1} \mid h_k) \\
        &\qquad\times \int \expect_* \left[ \left\{ \T_{k+1}( \{U_{*,j}, S_{*,j}, G_{*,j}\}_{j=k+1}^K)(O) - Q_{*,k+1}(H_{k+1}) \right\} V(O) \mid X>t_{k+1},H_{k+1}=(\ell_{k+1},h_k) \right] \\
        &\qquad\qquad P_{*,L_{k+1}}(\intd \ell_{k+1} \mid h_k,x_{[1:k]}=x^\dagger_{[1:k]},\delta_{[1:k]}=0) \qquad \text{(induction hypothesis)} \\
        &\quad+ S_{*,k}(t_{k+1} \mid h_k) \int Q_{*,k+1}(\ell_{k+1},h_k) V_{L_{k+1}}(\ell_{k+1} \mid h_k,x_{[1:k]}=x^\dagger_{[1:k]},\delta_{[1:k]}=0) \\
        &\qquad\qquad P_{*,L_{k+1}}(\intd \ell_{k+1} \mid h_k,x_{[1:k]}=x^\dagger_{[1:k]},\delta_{[1:k]}=0) \\
        &= U_{*,k}(h_k) \expect_* \left[ \left\{ \TS_{k,t_{k+1}}(S_{*,k},G_{*,k})(X,\Delta \mid h_k) - S_{*,k}(t_{k+1} \mid h_k) \right\} V(O) \mid X>t_k,H_k=h_k \right] \\
        &\quad+ S_{*,k}(t_{k+1} \mid h_k) \expect_* \Bigg[ \frac{\ind(X>t_{k+1})}{P_*(X>t_{k+1} \mid X>t_k,H_k=h_k)} \left\{ \T_{k+1}( \{U_{*,j}, S_{*,j}, G_{*,j}\}_{j=k}^K)(O) - Q_{*,k+1}(H_{k+1}) \right\} \\
        &\qquad\qquad\qquad\times V(O) \mid X>t_k,H_k=h_k \Bigg] \\
        &\quad+ S_{*,k}(t_{k+1} \mid h_k) \expect_* \Bigg[ \frac{\ind(X>t_{k+1})}{P_*(X>t_{k+1} \mid X>t_k,H_k=h_k)} \left\{ Q_{*,k+1}(H_{k+1}) - U_{*,k}(H_k) \right\} \\
        &\qquad\qquad\qquad\times V(O) \mid X>t_k,H_k=h_k \Bigg] \\
        &= \expect_*[f(O) V(O) \mid X>t_k,H_k=h_k]
    \end{align*}
    where $f(o)$ is defined as
    \begin{align*}
        & U_{*,k}(h_k) \left\{ \TS_{k,t_{k+1}}(S_{*,k},G_{*,k})(x,\delta \mid h_k) - S_{*,k}(t_{k+1} \mid h_k) \right\} \\
        &\quad+ S_{*,k}(t_{k+1} \mid h_k) \frac{\ind(x>t_{k+1})}{P_*(X>t_{k+1} \mid X>t_k,H_k=h_k)} \left\{ \T_{k+1}( \{U_{*,j}, S_{*,j}, G_{*,j}\}_{j=k+1}^K)(o) - Q_{*,k+1}(h_{k+1}) \right\} \\
        &\quad+ S_{*,k}(t_{k+1} \mid h_k) \frac{\ind(x>t_{k+1})}{P_*(X>t_{k+1} \mid X>t_k,H_k=h_k)} \left\{ Q_{*,k+1}(h_{k+1}) - U_{*,k}(h_k) \right\} \\
        &= \frac{\ind(x>t_{k+1})}{G_{*,k}(t_{k+1} \mid h_k)} \left\{ \T_{k+1}( \{U_{*,j}, S_{*,j}, G_{*,j}\}_{j=k+1}^K)(o) - U_{*,k}(h_k) \right\} \\
        &\quad+ U_{*,k}(h_k) \TS_{k,t_{k+1}}(S_{*,k},G_{*,k})(x,\delta \mid h_k) - Q_{*,k}(h_k) \\
        &= \T_{k}( \{U_{*,j}, S_{*,j}, G_{*,j}\}_{j=k}^K)(o) - Q_{*,k}(h_k).
    \end{align*}
    When simplifying $f$, we have used Condition~\ref{condition: cond ind cens} and \eqref{eq: SDR transform recursive}.
    Thus, \eqref{eq: Qk IF} holds for $k$, and so \eqref{eq: Qk IF} holds for all $k$ by induction.

    We now show the efficient influence function for $\alpha_*$.
    Since $\alpha_* = \expect_*[Q_{*,1}(L_1)]$, by \eqref{eq: Qk IF} at $k=1$ and a similar calculation as above,
    \begin{align*}
        \left. \frac{\partial \alpha_\epsilon}{\partial \epsilon} \right|_{\epsilon=0} &= \left. \frac{\partial}{\partial \epsilon} \int Q_{\epsilon,1}(\ell_1) \{1+\epsilon V_{L_1}(\ell_1) \} P_{*,L_1}(\intd \ell_1) \right|_{\epsilon=0} \\
        &= \expect_* \Big[ \expect_*[ \{ \T_1(\{U_{*,j}, S_{*,j}, G_{*,j}\}_{j=1}^K)(O) - Q_{*,1}(L_1) \} V(O) \mid L_1] \Big] + \expect_*[ \{ Q_{*,1}(L_1) - \alpha_* \} V(O) ] \\
        &= \expect_*[D(\{U_{*,j}, S_{*,j}, G_{*,j}\}_{j=1}^K, \alpha_*) V(O)],
    \end{align*}
    where the last step follows from the tower rule.
    Since $V \in \mathcal{V}$ is arbitrary, the closure of $\mathcal{V}$ is the tangent set $L_0^2(P_*)$, and $\expect_*[D(\{U_{*,j}, S_{*,j}, G_{*,j}\}_{j=1}^K, \alpha_*)]=0$ (by Theorem~\ref{thm: SDR identify}), we have that $D(\{U_{*,j}, S_{*,j}, G_{*,j}\}_{j=1}^K, \allowbreak \alpha_*)$ is the canonical gradient of $P \mapsto \alpha_P$ at $P_*$ under a nonparametric model.
\end{proof}

The following lemma shows that the remainder is negligible.

\begin{lemma} \label{lemma: negligible remainder}
    Under Conditions~\ref{condition: cond ind cens}, \ref{condition: consistency of strong positivity}, \ref{condition: negligible remainder} and \ref{condition: bounded U}, $\| R_k(\{\hat{U}_j^\SDR, \hat{S}_j, \hat{G}_j\}_{j=k}^K \mid \cdot) \|_k = \smallo_p(n^{-1/2})$ for all $k=1,\ldots,K$, $\| \bar{U}_k^\SDR - U_{*,k} \|_{k+1} = \smallo_p(n^{-1/2})$ for all $k=1,\ldots,K-1$, and $P_* \T_{1}( \{\hat{U}_j^\SDR, \hat{S}_j, \hat{G}_j\}_{j=1}^K) - \alpha_* = \smallo_p(n^{-1/2})$.
\end{lemma}

\begin{proof}
    We prove by induction. Since $R_{K}(\{\hat{U}_K^\SDR,\hat{S}_K,\hat{G}_K\} \mid H_K) = \bar{R}_{K,t_{K+1}}(\hat{S}_K,\hat{G}_K \mid H_K)$, Condition~\ref{condition: negligible remainder} implies that $\| R_{K}(\{\hat{U}_K^\SDR,\hat{S}_K,\hat{G}_K\} \mid \cdot) \|_K = \smallo_p(n^{-1/2})$. By \eqref{eq: SDR remainder}, $\| \bar{U}_{K-1}^\SDR - U_{*,K-1} \|_{K} \leq \| R_{K}(\{\hat{U}_K^\SDR,\hat{S}_K,\hat{G}_K\} \mid \cdot) \|_K = \smallo_p(n^{-1/2})$.

    Suppose that $\| R_{k+1}(\{\hat{U}_j^\SDR, \hat{S}_j, \hat{G}_j\}_{j=k+1}^K \mid \cdot) \|_{k+1} = \smallo_p(n^{-1/2})$, so that $\| \bar{U}_{k}^\SDR - U_{*,k} \|_{k+1} = \smallo_p(n^{-1/2})$. By Lemma~\ref{lemma: k k+1 norms}, \eqref{eq: SDR remainder recursive}, Conditions~\ref{condition: consistency of strong positivity}, \ref{condition: negligible remainder} and \ref{condition: bounded U}, and Cauchy-Schwarz inequality,
    \begin{align*}
        & \| R_k(\{\hat{U}_j^\SDR, \hat{S}_j, \hat{G}_j\}_{j=k}^K \mid \cdot) \|_k \\
        &\lesssim \| R_{k+1}(\{\hat{U}_j^\SDR, \hat{S}_j, \hat{G}_j\}_{j=k+1}^K \mid \cdot) \|_{k+1} + \| \hat{G}_{k} - G_{*,k} \|_{(t_{k},t_{k+1}],k} \| \hat{U}_{k}^\SDR - U_{*,k} \|_{k+1} + \| \bar{R}_{k,t_{k+1}}(\hat{S}_k,\hat{G}_k \mid \cdot) \|_k \\
        &\leq \| R_{k+1}(\{\hat{U}_j^\SDR, \hat{S}_j, \hat{G}_j\}_{j=k+1}^K \mid \cdot) \|_{k+1} + \| \hat{G}_{k} - G_{*,k} \|_{(t_{k},t_{k+1}],k} \| \hat{U}_{k}^\SDR - \bar{U}_{k}^\SDR \|_{k+1} \\
        &\quad+ \| \hat{G}_{k} - G_{*,k} \|_{(t_{k},t_{k+1}],k} \| \bar{U}_{k}^\SDR - U_{*,k} \|_{k+1} + \| \bar{R}_{k,t_{k+1}}(\hat{S}_k,\hat{G}_k \mid \cdot) \|_k \\
        &= \smallo_p(n^{-1/2}).
    \end{align*}
    By induction, $\| R_k(\{\hat{U}_j^\SDR, \hat{S}_j, \hat{G}_j\}_{j=k}^K \mid \cdot) \|_k = \smallo_p(n^{-1/2})$ for all $k=1,\ldots,K$, and $\| \bar{U}_k^\SDR - U_{*,k} \|_{k+1} = \smallo_p(n^{-1/2})$ for all $k=1,\ldots,K-1$.
    By \eqref{eq: SDR remainder}, $|P_* \T_{1}( \{\hat{U}_j^\SDR, \hat{S}_j, \hat{G}_j\}_{j=1}^K) - \alpha_*| = | P_* [\T_{1}( \{\hat{U}_j^\SDR, \hat{S}_j, \hat{G}_j\}_{j=1}^K) - Q_{*,1}] | \leq \| R_1(\{\hat{U}_j^\SDR, \hat{S}_j, \hat{G}_j\}_{j=1}^K \mid \cdot) \|_1 = \smallo_p(n^{-1/2})$.
\end{proof}

\begin{proof}[Proof of Theorem~\ref{thm: efficiency}: Efficiency of $\hat{Q}_0^\SDR$]
    In this proof, we show that $\hat{Q}_0^\SDR$ is asymptotically efficient under a nonparametric model.
    \begin{align*}
        \hat{Q}_0^\SDR - \alpha_* &= P_n \T_1(\{\hat{U}_{j}^\SDR, \hat{S}_{j}, \hat{G}_{j}\}_{j=1}^K) - P_* \T_1(\{U_{*,j}, S_{*,j}, G_{*,j}\}_{j=1}^K) \\
        &= (P_n-P_*) \T_1(\{U_{*,j}, S_{*,j}, G_{*,j}\}_{j=1}^K) + P_* \T_1(\{\hat{U}_{j}^\SDR, \hat{S}_{j}, \hat{G}_{j}\}_{j=1}^K) - \alpha_* \\
        &\quad+ (P_n-P_*) \left\{ \T_1(\{\hat{U}_{j}^\SDR, \hat{S}_{j}, \hat{G}_{j}\}_{j=1}^K) - \T_1(\{U_{*,j}, S_{*,j}, G_{*,j}\}_{j=1}^K) \right\} \\
        &= P_n D(\{U_{*,j}, S_{*,j}, G_{*,j}\}_{j=1}^K, \alpha_*) + \smallo_p(n^{-1/2}),
    \end{align*}
    where the last step follows from Lemma~\ref{lemma: negligible remainder} and Conditions~\ref{condition: IF consistency}--\ref{condition: Donsker} \protect\citepsupp[see, e.g., Eq.~2.1.8 in][]{vandervaart1996}.
    Hence, $\hat{Q}_0^\SDR$ is an asymptotically linear estimator of $\alpha_*$ with influence function $D(\{U_{*,j}, S_{*,j}, G_{*,j}\}_{j=1}^K, \alpha_*)$ and is thus asymptotically efficient.
\end{proof}

\subsection{Product bias remainder for discrete times}

\begin{proof}[Proof of Lemma~\ref{lemma: equivalent difference for hazard and survival functions}]
    We prove the two inequalities separately. We prove the first inequality by induction on $j$. When $j=1$, we have that $|\hat{S}(t^j \mid h) - S_*(t^j \mid h)| = | [1-\hat{\lambda}(t^j \mid h)] -  [1-\lambda_*(t^j \mid h)]| = | \hat{\lambda}(t^j \mid h) - \lambda_*(t^j \mid h) |$ and the desired inequality holds trivially. Suppose that the desired inequality holds for $j-1$, then it holds that
    \begin{align}
        & \hat{S}(t^j \mid h) - S_*(t^j \mid h) \nonumber \\
        &= \hat{S}(t^{j-1} \mid h) [1-\hat{\lambda}(t^j \mid h)] - S_*(t^{j-1} \mid h) [1-\lambda_*(t^j \mid h)] \nonumber \\
        &= [S_*(t^{j-1} \mid h) + \hat{S}(t^{j-1} \mid h) - S_*(t^{j-1} \mid h)] [1-\hat{\lambda}(t^j \mid h)] - S_*(t^{j-1} \mid h) [1-\lambda_*(t^j \mid h)] \nonumber \\
        &= [\hat{S}(t^{j-1} \mid h) - S_*(t^{j-1} \mid h)] [1-\hat{\lambda}(t^j \mid h)] + S_*(t^{j-1} \mid h) [\lambda_*(t^j \mid h) - \hat{\lambda}(t^j \mid h)]. \label{eq: hazard and survival recursive}
    \end{align}
    Therefore, by the induction hypothesis and the fact that $\hat{\lambda}(t^j \mid h) \in [0,1]$,
    \begin{align*}
        & \left| \hat{S}(t^j \mid h) - S_*(t^j \mid h) \right| \\
        &\leq \left| \hat{S}(t^{j-1} \mid h) - S_*(t^{j-1} \mid h) \right| [1-\hat{\lambda}(t^j \mid h)] + S_*(t^{j-1} \mid h) \left| \hat{\lambda}(t^j \mid h) - \lambda_*(t^j \mid h) \right| \\
        &\lesssim \sup_{t \in (0,t_2]} \left| \hat{\lambda}(t \mid h_k) - \lambda_*(t \mid h_k) \right|,
    \end{align*}
    as desired.
    
    We now prove the second inequality. Rearranging \eqref{eq: hazard and survival recursive},
    \begin{align*}
        &\hat{\lambda}(t^j \mid h) - \lambda_*(t^j \mid h) \\
        &= \frac{1}{S_*(t^{j-1} \mid h)} \Bigg\{ [\hat{S}(t^{j-1} \mid h) - S_*(t^{j-1} \mid h)] [1-\hat{\lambda}(t^j \mid h)] - [\hat{S}(t^j \mid h) - S_*(t^j \mid h)] \Bigg\}.
    \end{align*}
    Since $S_*(t^{j-1} \mid h) \geq S_*(t_2 \mid h_k) \geq \epsilon > 0$ for some constant $\epsilon$, the desired inequality holds.
\end{proof}

\section{Additional results in COVID-19 vaccine trial analysis} \label{sec: vac trial2}

The log multiplicative CDE is defined as $\log\CDE(w_a,w_c) := \log\{1 \allowbreak- \beta_{*}(w_a,1,w_c)\} \allowbreak- \log \{1 \allowbreak- \beta_{*}(w_a,0,w_c)\}$. Due to the low incidence and the lack of participants with high antibody levels, particularly for $\tau \leq 90$, $\log\CDE$ estimates are unstable 
for relatively large values of
Day~22 antibody biomarker levels $w_a$.

\begin{figure}[h!tb]
    \centering
    \includegraphics[width=\linewidth]{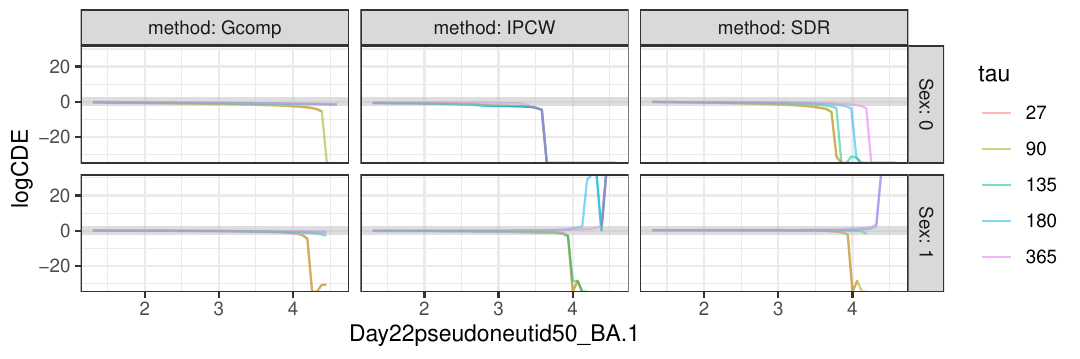}
    \caption{Similar to Fig.~\ref{fig: CDE nAb} for log multiplicative CDE of Day~22 neutralizing antibody level (\texttt{Day22pseudoneutid50\textunderscore BA.1}).}
    \label{fig: logCDE nAb}
\end{figure}

\begin{table}[h!tb]
    \centering
    \caption{Similar to Table~\ref{tab: marg nAb} accounting for informative censoring due to spike binding antibody biomarker levels at Days~22 and 43, instead of neutralizing antibody biomarker levels. All estimates and standard errors are multiplied by 100.}
    \begin{tabular}{c|r|r|r|r|r}
    \hline\hline
         Time $\tau$ & 27 & 90 & 135 & 180 & 365 \\
         \hline
         Est (S.E.) & -0.60 (0.25) & -0.91 (0.35) & -1.14 (0.48) & -0.86 (0.55) & -0.58 (0.67) \\
    \hline\hline
    \end{tabular}
    \label{tab: marg bAb}
\end{table}

\begin{figure}[h!tb]
    \centering
    \includegraphics[width=\linewidth]{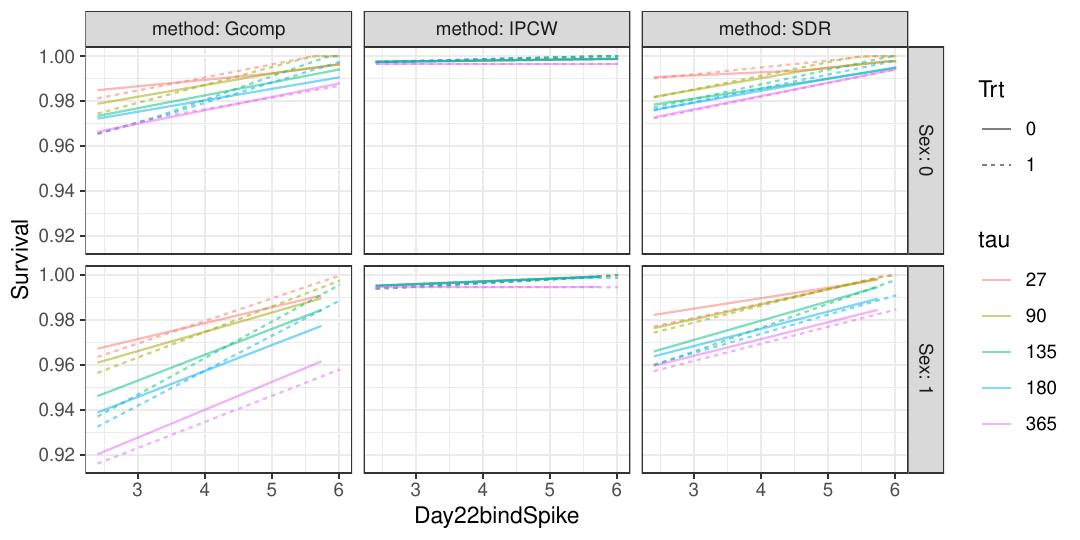}
    \caption{Similar to Fig.~\ref{fig: surv nAb} for Day~22 spike binding antibody level (\texttt{Day22bindSpike}).}
    \label{fig: surv bAb}
\end{figure}

\begin{figure}[h!tb]
    \centering
    \includegraphics[width=\linewidth]{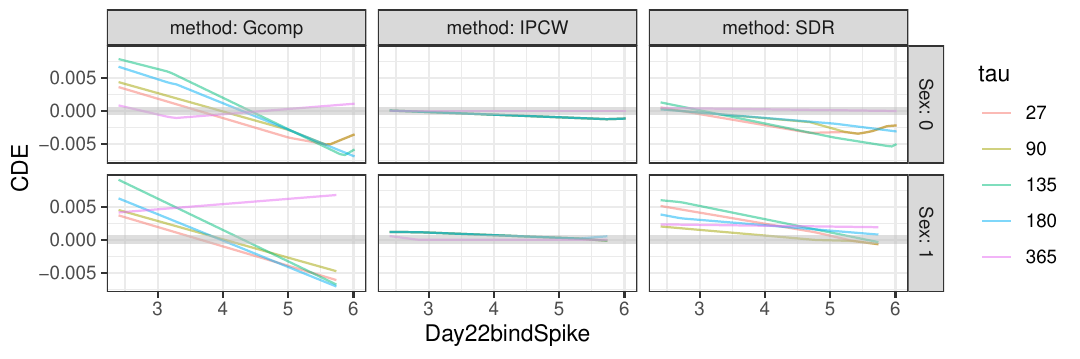}
    \caption{Similar to Fig.~\ref{fig: CDE nAb} for Day~22 spike binding antibody level (\texttt{Day22bindSpike}).}
    \label{fig: CDE bAb}
\end{figure}

\begin{figure}[h!tb]
    \centering
    \includegraphics[width=\linewidth]{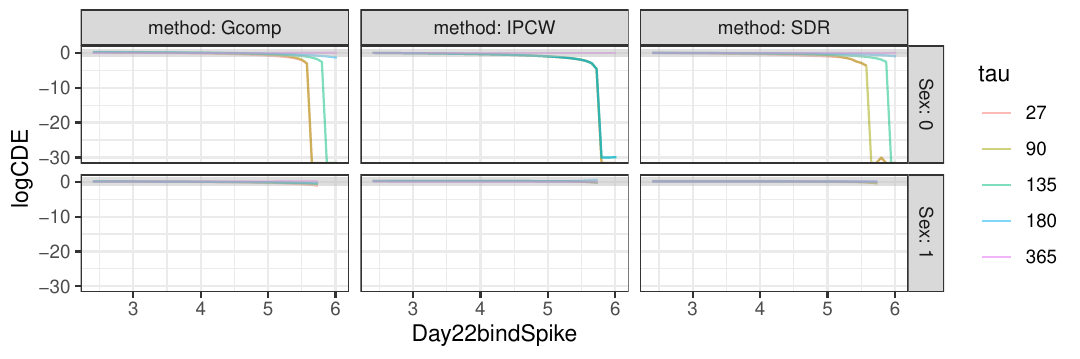}
    \caption{Similar to Fig.~\ref{fig: logCDE nAb} for Day~22 spike binding antibody level (\texttt{Day22bindSpike}).}
    \label{fig: logCDE bAb}
\end{figure}

\begin{figure}[h!tb]
    \centering
    \includegraphics[width=\linewidth]{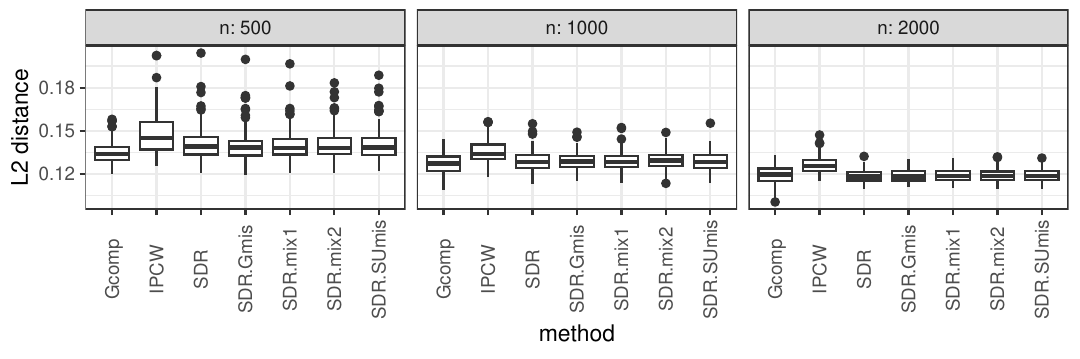}
    \caption{Similar to Fig.~\ref{fig: disc condp L2} without \texttt{Gcomp.SUmis} and \texttt{IPCW.Gmis}.}
    \label{fig: disc condp L22}
\end{figure}

\clearpage

\bibliographystylesupp{abbrvnat}
\bibliographysupp{reference}

\end{document}